\documentclass[review,3p,times]{elsarticle}

\usepackage{array}
\usepackage{verbatim}
\usepackage{bm}
\usepackage{multirow}
\usepackage{amsmath}
\usepackage{amssymb}
\usepackage{graphicx}
\usepackage{epstopdf}
\usepackage{tikz-cd}
\usepackage{esint}
\usepackage{comment}
\usepackage{float}
\usepackage[subfigure]{graphfig}
\usepackage[unicode=true,
 bookmarks=false,
 breaklinks=false,pdfborder={0 0 1},colorlinks=false]
 {hyperref}
\usepackage{tikz}
\usetikzlibrary{shapes.geometric, arrows, positioning, matrix}

\makeatletter

\usepackage{tabularx}\usepackage{subfigure}\usepackage{subfigure}\usepackage{color}\usepackage{textcomp}\usepackage{latexsym}

\numberwithin{equation}{section}

\usepackage{setspace}

\usepackage{appendix}

\usepackage{bm}
\usepackage{amsthm}

\newtheorem*{proposition*}{\bf Proposition}

\usepackage{algorithm}
\usepackage{algpseudocode}

\journal{}

\makeatother

\color{black}

\graphicspath{}
\begin{document}

\title{A novel hybrid approach for accurate simulation of compressible multi-component flows across all-Mach number }
\author[ad1]{Xi Deng \corref{cor}}

\author[ad2]{Bin Xie}

\author[ad1]{Omar K. Matar}

\author[ad3]{Pierre Boivin}


\address[ad1]{Department of Chemical Engineering, Imperial College London, SW7 2AZ, United Kingdom}

\address[ad2]{School of Naval Architecture, Department of Ocean and Civil Engineering, Shanghai Jiaotong University, Shanghai, 200240, China}

\address[ad3]{Aix Marseille Univ, CNRS, Centrale Med, M2P2, Marseille, France}

\cortext[cor]{Corresponding author: x.deng@imperial.ac.uk}

\begin{abstract}
Numerical simulation of multi-component flow systems characterized by the simultaneous presence of pressure-velocity coupling and pressure-density coupling dominated regions remains a significant challenge in computational fluid dynamics. Thus, this work presents a novel approach that combines the Godunov-type scheme for high-speed flows with the projection solution procedure for incompressible flows to address this challenge. To simulate compressible multi-component flows, this study employs the homogeneous model and solves its conservative form by the finite-volume method, which enables the application of the one-fluid model solution procedure while satisfying conservation. The proposed hybrid approach begins by splitting the inviscid flux into the advection part and the pressure part. The solution variables are first updated to their intermediate states by solving the advection part with the all-speed AUSM (Advection Upwind Splitting Method) Riemann solver. The advection flux in AUSM is modified to eliminate the pressure flux term that deteriorates the accuracy at the low Mach region. To prevent the advection flux from causing spurious velocities when surface tension is present, the pressure-velocity coupling term is modified to ensure it vanishes at material interfaces. Then, we derive the pressure Helmholtz equation to solve the final pressure and update the intermediate states to the solution variables at the next time step. The proposed hybrid approach retains the upwind property of the AUSM scheme for high Mach numbers while recovering central schemes and the standard projection solution for low Mach limits. We also prove that the proposed hybrid method is capable of solving the stationary contact discontinuity exactly. To accurately resolve the complex flow structures including shock waves and material interfaces without numerical oscillations, a newly proposed homogenous ROUND (Reconstruction Operator on Unified Normalised-variable Diagram) reconstruction strategy is employed in this work.  By simulating high-speed compressible multiphase flows and incompressible multiphase flows, this study demonstrates the ability of the proposed method to accurately handle flow regimes across all Mach numbers. Its capability to address multi-physical processes at all Mach numbers is further validated through simulations of nucleate boiling and the Richtmyer-Meshkov instability with cavitation.

\end{abstract}

\maketitle

\newpage
\section{Introduction}
Multi-component flows with moving material interfaces are common in various engineering applications and occur in both compressible and incompressible flow regimes. Examples such as underwater explosions and the combustion chambers of hypersonic vehicles involve highly compressible flows characterized by interactions between shock waves and material interfaces, where pressure-density coupling is dominant. In contrast, examples such as water tanks are characterized by low-speed flows in which pressure-velocity coupling is dominant. In some cases such as the lithium boiling tank in the cooling system of the nuclear reactors, the flow systems are characterized by the simultaneous presence of pressure-velocity coupling and pressure-density coupling dominated regions. The numerical simulation of such a flow system requires a model capable of accurately handling multiphase flows across all Mach numbers.

Numerical simulations of multiphase flows are generally classified into two approaches based on the treatment of material interfaces.
One of the approaches is the sharp interface method (SIM). SIM models the material interface as a sharp discontinuity. The interface is explicitly represented using techniques such as the front-tracking method \cite{tryggvason2001front}, the level-set function \cite{osher1988fronts}, or the ghost fluid method \cite{fedkiw1999non}. These methods excel in preserving a sharply resolved interface with high accuracy. However, challenges arise when the interface becomes highly wrinkled, particularly for the front-tracking method. For the level-set and ghost fluid methods, ensuring robustness and maintaining mass conservation require careful consideration. Additionally, handling dynamic creation and disappearance of interfaces due to phase transitions is not straightforward for SIM. In contrast, another approach, known as the diffuse interface method (DIM), treats the material interface as a diffusive zone. One of the drawbacks of DIM is that the numerical interface thickness can be substantially larger than the actual physical thickness. However, DIM offers advantages, such as its ability to straightforwardly handle large distortions of the material interface and model phase transitions. Furthermore, as pointed out by \citet{lemartelot2013liquid,demou2022pressure},
DIM provides additional advantages for compressible multiphase simulations, including the retention of thermodynamic consistency everywhere and the ability to robustly calculate property gradients due to the diffused interface.

Several models are available for modelling compressible multiphase flow systems using the DIM approach. The seven-equation model, also known as the Baer-Nunziato model \cite{baer1986two}, is the most general model. This model and its variant \cite{saurel1999multiphase} consist of two momentum equations, two energy equations, equations for the mass conservation in each phase, and the evolution equation for the volume fraction. In this model, each component is described by a complete set of parameters, which means that kinetic, mechanical, and thermal equilibria are not enforced within each cell state. Although the seven-equation model is a physically complete model, numerically solving the seven-equation model is complicated, since the model involves complex wave structures and stiff relaxation processes. Derived from the seven-equation model, several reduced models have been proposed through relaxation processes that force the system to achieve specific equilibrium states. For example, the six-equation model, which enforces kinetic equilibrium, has been developed and improved by \citet{saurel2009simple,pelanti2014mixture}. Through enforcing kinetic and mechanical equilibria in each cell state, the five-equation model, originally known as the Kapila model \cite{kapila2001two},  has been developed and studied by \citet{allaire2002five,murrone2005five,jain2020conservative}. Recent work by \citet{zhang2023mathematical} derives a five-equation model by assuming that the pressure and the velocity relaxation have a larger time scale than the original Kapila model. The five-equation model has been considered a good compromise between physical completeness and numerical complexity. 

The four-equation model, also known as the homogenous model, is derived by assuming kinetic, mechanical, and thermal equilibria\cite{abgrall1996prevent,saurel1999simple,johnsen2012preventing}. One drawback of the four-equation model is the emergence of pressure and velocity oscillations across material interfaces when the fully-conservative form is solved using the Godunov-type finite-volume method. Thus, special techniques, such as the non-conservative double-flux strategy \cite{abgrall2001computations}, or the inclusion of additional equations for thermodynamic parameters \cite{shyue1998efficient,shyue1999fluid}, are required to prevent model-induced oscillations. However, the four-equation model offers the advantage of conveniently handling an arbitrary number of components. Furthermore, when heat conductivity is included and the thermal boundary layers at the material interface are properly resolved, the thermal equilibrium assumption becomes appropriate \cite{demou2022pressure}. Thus, using the four-equation model for boiling is justified \cite{le2014towards}. Since our target is to simulate the lithium boiling tank, the four-equation model will be adopted in this work. 

The development of truly all-Mach methods for compressible multiphase flow simulations remains a challenge. Although density-based solvers are widely employed for high-speed flows, they encounter stiffness issues and reduced accuracy when applied to low-Mach number flows. To overcome this issue, the works of \citet{le2014towards,murrone2008behavior,pelanti2017low} adopt preconditioning techniques to simulate multiphase flows at low-Mach number using the density-based solver. However, as noted by \citet{turkel2005local,saurel2018diffuse}, there remains much room for improving the efficiency of preconditioning methods in simulating unsteady and multi-component flows. Another approach to simulating compressible multiphase flows using density-based solvers involves modifying the numerical flux via the Riemann solver. For instance, this strategy was adopted by \citet{jiao2024all}. As discussed by \citet{jiao2024all}, a potential drawback of this strategy is that it may not achieve the same level of accuracy as the original Riemann solver for high-speed flow simulations. On the other hand, the pressure-based solver, originally developed for incompressible flows, has been extended to simulate high-speed compressible multiphase flows, as demonstrated by \citet{jemison2014compressible,denner2018pressure}. 

Pressure-based solvers for all-Mach number flows have been further developed by \citet{dumbser2016conservative,park2005multiple,bermudez2020staggered,busto2021semi,xie2019high}. Although these pressure-based all-Mach-number solvers are robust in handling simulations, the results show that their accuracy for high-speed flows is often unsatisfactory. Therefore, there remains significant potential for developing methods that are both accurate and robust across all Mach numbers. Recent work by \citet{demou2022pressure} introduced a novel pressure-based formulation for four-equation models, enabling the simulation of low-Mach multiphase flows with phase-transition phenomena such as nucleate boiling. Since the mentioned work solves the non-conservative form, their proposed formulation may not be directly applied to scenarios involving shock waves, such as lithium boiling tanks whose simulation is the motivation of the current study. 

When simulating compressible multiphase flows using DIM, it is crucial to prevent excessive diffusion of the material interface, particularly in long-term simulations. The work of \cite{johnsen2006implementation,de2015new} applied the high-order weighted essentially non-oscillatory (WENO) schemes \cite{liu1994weighted} or discontinuous Galerkin methods to simulate compressible multiphase flows. To achieve a sharp resolution of material interfaces, compressive numerical schemes \cite{shukla2010interface,chiapolino2017sharpening} and anti-diffusion methods \cite{so2012anti} have been developed. A low-dissipation reconstruction strategy, known as the Boundary Variation Diminishing (BVD) scheme, has been applied to solve the five-equation model \cite{deng2018high} and has been further extended to unstructured grids \cite{cheng2021low}. However, as shown by \citet{cheng2022new}, both WENO and BVD schemes are computationally expensive when applied to unstructured grids. Recently, a high-resolution, scalar-structure-preserving scheme named ROUND (reconstruction operator on unified-normalised-variable diagram) has been developed in a framework that unifies existing non-linear schemes in a compact stencil \cite{deng2023unified}. The numerical tests demonstrate that ROUND schemes produce superior resolution than the fifth-order WENO scheme across the discontinuity. The work \cite{deng2023new,deng2023large} extended the ROUND scheme to unstructured grids and demonstrated its efficacy compared to conventional TVD (total variation diminishing) schemes. Therefore, applying ROUND schemes to simulate compressible multiphase flows holds significant potential.   

This work presents a novel approach that combines the density-based solver for high-speed flows with the projection solution procedure for incompressible flows to simulate compressible multiphase flows across all-Mach number. The homogeneous model and its conservative form are solved. Thus, the conservation law is strictly satisfied. The proposed hybrid approach retains the upwind property of the Godunov-type scheme for high Mach numbers while retrieving central schemes and the standard projection solution for low Mach limits. We also prove that the proposed hybrid method is capable of solving the stationary contact discontinuity exactly. To achieve high-resolution and oscillation-free simulation, a newly proposed homogenous ROUND reconstruction strategy is employed in this work.  By simulating high-speed compressible multiphase flows and incompressible multiphase flows with and without phase transition, this study demonstrates the ability of the proposed method to accurately handle complex compressible multicomponent flows across all Mach numbers.

The paper is organized as follows. First, we introduce the governing equations and the associated thermodynamic closures in Section \ref{sec:numericalModel}. The details of the proposed hybrid approach for all-Mach flows are given in Section \ref{sec:solutionProcedures}. Then, Section \ref{sec:multiphysics} presents the details of modelling surface tension and phase transition. In Section \ref{sec:numericalMethods}, the numerical methods employed in this work will be given. The performance of the proposed solver will be validated through numerical experiments in Section \ref{sec:numericalExperiments}. Conclusions and perspectives are drawn in Section \ref{sec:concluding}.

\section{Numerical models for multi-component systems} \label{sec:numericalModel}

\subsection{Governing equations}
A homogeneous mixture model also known as the four-equation model is employed in this work. The homogeneous model assumes kinetic, mechanical, and thermal equilibria between different components, which means that pressure, temperature, and velocity are shared by all components within a single computational cell. Detailed discussions of the model's underlying assumptions can be found in \cite{kapila2001two, saurel1999multiphase, abgrall2001computations}. As part of the diffuse interface methods, the homogeneous mixture model treats the material interface as a diffuse zone, allowing the same numerical method to be applied across different thermodynamic states (liquid/gaseous). This model can naturally handle the creation and destruction of material interfaces due to phase transitions \cite{saurel2001multiphase}. In this work, quantities for individual components are indicated with the subscript $k$, while mixture quantities are denoted without a subscript. 

The homogenous model including viscosity, heat conductivity, surface tension effect, and gravity  is given by
\begin{equation}
\begin{cases}
  \begin{aligned}
 &\frac{\partial \rho Y_k}{\partial t} +  \mathbf{\nabla} \cdot( \rho Y_k \textbf{u}  ) = \dot\omega_{k}, \\
 &\frac{\partial \rho \textbf{u}}{\partial t} + \mathbf{\nabla} \cdot \Bigl(\rho \textbf{u}\otimes\textbf{u} + p\overline{\overline{I}} \Bigl) = \mathbf{\nabla} \cdot( \overline{\overline{\tau}})+\mathbf{f_{g}}+\mathbf{f_{s}}, \\
 & \frac{\partial \rho E}{\partial t} +\mathbf{\nabla} \cdot \Bigl( (\rho E + p)\textbf{u} \Bigl)  = \mathbf{\nabla} \cdot( \overline{\overline{\tau}} \cdot \textbf{u})-\mathbf{\nabla} \cdot(\textbf{q})+(\mathbf{f_{g}}+\mathbf{f_{s}}) \cdot \textbf{u}, 
  \end{aligned}
  \end{cases}
   \label{eq_PDE}
\end{equation}
where $\rho$, $\textbf{u}$, and $p$ are the density, the velocity vector, and the pressure of the mixture, respectively. The mass fraction for each component is denoted as $Y_{k}$. The mass transfer for each component $k$ due to phase transition is $\dot\omega_{k}$. $E$ is the mixture total energy and defined as $E=e+\frac{1}{2}\textbf{u}^{2}$, where $e$ is the specific mixture internal energy. The viscous stress tensor $\overline{\overline{\tau}}$ is introduced to account for the viscous effect, expressed as follows
\begin{equation}
\overline{\overline{\tau}}= \mu \Bigl( \mathbf{\nabla}\textbf{u}+(\mathbf{\nabla}\textbf{u})^{\text{T}}-\frac{2}{3} \mathbf{\nabla} \cdot(\textbf{u})\overline{\overline{I}} \Bigl),
\end{equation}
where $\mu$ is the mixture dynamic viscosity. The energy flux $\textbf{q}$ is calculated by Fourier's law as 
\begin{equation}
\textbf{q}=-\lambda \cdot \mathbf{\nabla}T,
\label{eq:heatFlux}
\end{equation}
where $T$ is the mixture temperature and $\lambda$ is the mixture heat conduction coefficient. Following \cite{le2014towards}, the mixture dynamic viscosity and heat conduction are calculated by the mixture rule as
\begin{equation}
    \mu=\sum_{k=1}^N \alpha_k \mu_k,~~\text{and}~~\lambda=\sum_{k=1}^N \alpha_k \lambda_k,
\end{equation}
where $\alpha_k$ is the volume fraction for each component. $\mathbf{f_{g}}$ is the gravity force and calculated as $\mathbf{f_{g}}=\rho \mathbf{g}$, where $\mathbf{g}$ is the acceleration due to gravity. Finally, $\mathbf{f_{s}}$ is the surface tension force. The details of calculating this force will be given in the following sections. It is noted that mass conservation equations can also be solved alternatively as
\begin{equation}
\frac{\partial \rho_k \alpha_k}{\partial t} +  \mathbf{\nabla} \cdot( \rho_k \alpha_k \textbf{u}  ) =\dot\omega_{k},
\end{equation}
where $\rho_k$ is the material density of each component. 

\subsection{Equation of state}
The equation of state (EOS) which provides a pressure law $p_k(e_k,\rho_k)$ and a temperature law $T_k(p_k,\rho_k)$ is required to close the equation system. In this work, we use Noble-Abel Stiffened Gas (NASG) \cite{le2016noble} for thermodynamic closure. Compared with Stiffened Gas (SG) EOS, NASG improves the liquid specific volume accuracy by taking the repulsive molecular effects into account. NASG EOS for a given component is expressed as
\begin{equation}
\begin{cases}
  \begin{aligned}
 & p_k=\frac{(\gamma_k-1)\rho_k}{1-\rho_k b_k} (e_k-q_k)-\gamma_k p_{\infty,k}, \\
 & T_k=\frac{(1-\rho_k b_k)(p_k+p_{\infty,k})}{C_{v,k}\rho_k(\gamma_k-1)},\\
 &c_k= \sqrt{\gamma_k \frac{p_k+p_{\infty,k}}{\rho_k (1-\rho_k b_k)}},
  \end{aligned}
  \end{cases}
   \label{eq:EOS}
\end{equation}
where $c_k$ is the sound speed of pure component and  parameter $C_{v,k}$ is component's specific heat capacity at constant pressure. $\gamma_k$, $b_k$, $q_k$ and $p_{\infty,k}$ are other parameters of EOS for each component. It is noted that $\gamma_k=\frac{C_{p,k}}{C_{v,k}}$, where $C_{p,k}$ is the component's specific heat capacity. All parameters are assumed to be constant in this work.  

The EOS for the mixture is determined based on the four-equation model, which assumes temperature and pressure equilibrium, along with relations for mixture density and mixture energy:
\begin{equation}
\begin{cases}
T=T_k,~\forall k \\
p=p_k,~\forall k\\
v =\sum^{N}_{k=1}{Y_k v_k}, \\
e = \sum^{N}_{k=1}{Y_k e_k}, 
\end{cases}
\label{eq:equilibrium}
\end{equation}
where $v$ is specific volume. This work assumes that only one component may be found in its liquid form and is given index $k=1$. Then, there is $p_{\infty,k}=0$ and $b_k=0$, $\forall k>1$. Then, by solving Eq.~\ref{eq:EOS} and imposing the conditions give by Eq.~\ref{eq:equilibrium}, the mixture temperature is
\begin{equation}
T
= \frac{e-\sum^N_{k=1} Y_k q_k}{ \sum^N_{k=1} Y_k C_{v,k}  \left( \frac{p+ \gamma_k p_{\infty, k}}{p+p_{\infty, k}} \right) }, \\
\end{equation}
and the mixture pressure is obtained as
\begin{equation}
 p = \frac{\tilde{b}+\sqrt{\tilde{b}^2 + 4 \tilde{a} \tilde{c}}}{2\tilde{a}},
\end{equation}
with
 \begin{equation}
 \begin{cases}
  \begin{aligned}
&\tilde{a}=\bar{C_{v}}, \\
&\tilde{b}=   \left( \frac{e-\bar{q}}{v-\bar{b}}  \right) \left(  \bar{C_{p}}-\bar{C_{v}} \right)-p_{\infty, 1}\bar{C_{v}} - p_{\infty, 1}Y_1 \left( C_{p,1}-C_{v,1} \right)   ,    \\
&\tilde{c}=    \left(\frac{e-\bar{q}}{v-\bar{b}} \right) p_{\infty, 1} \left[  \bar{C_{p}}-\bar{C_{v}} - Y_1 \left( C_{p,1}-C_{v,1} \right)   \right],
  \end{aligned}
 \end{cases} \label{eq_eos_p_coeff}
\end{equation}
and
 \begin{equation}
  \bar{C_{v}}=\sum^N_{k=1} Y_k C_{v,k},  \hspace{0.5cm} 
  \bar{C_{p}}=\sum^N_{k=1} Y_k C_{p,k},  \hspace{0.5cm} 
  \bar{q} = \sum^N_{k=1} Y_k q_k,  \hspace{0.5cm} 
  \bar{b} = \sum^N_{k=1} Y_k b_k.
  \label{eq_coeffEOS}
\end{equation}


\section{Solution procedures for all Mach number flow} \label{sec:solutionProcedures}
The key solution procedure for solving the equation systems is to update the inviscid flux term. Thus, we separate the inviscid flux term from the remainder of the control equation system. The equation system with only the inviscid flux term is written as
\begin{equation} \label{eq:invEquation}
   \frac{\partial \mathbf{U}}{\partial t} +  \mathbf{\nabla} \cdot \mathbf{F}_{inv}(\mathbf{U}) =0,
\end{equation}
where solution variables $\mathbf{U}=\Bigl(\rho Y_1,...\rho Y_k, \rho \mathbf{u}, \rho E\Bigl)^{\text{T}}$ and inviscid flux is $\mathbf{F}_{inv}=\Bigl(\rho Y_1 \textbf{u},...\rho Y_k \textbf{u}, \rho \textbf{u}\otimes\textbf{u} + p\overline{\overline{I}}, (\rho E + p)\textbf{u}\Bigl)^{\text{T}}$. To numerically solve the above equation system, there are generally two categories, namely, density-based solvers and pressure-based solvers. Density-based solvers are suitable for high-speed flows where the pressure-density coupling is dominant. However, density-based solvers such as Godunov-type schemes generally exhibit deteriorating accuracy for low-Mach number flow. Introducing preconditioning techniques is one of the solutions to overcome the limitations at low Mach number. However, as pointed out by \citet{saurel2018diffuse}, preconditioning techniques need to be more efficient for unsteady flow and multiphase flow simulations. On the other hand, pressure-based solvers have advantages in solving low-Mach number flows where the pressure-velocity coupling dominates. However, when solving the non-conservative form, the numerical solution of pressure-based solvers for strong shock waves may not be correct. For compressible multiphase flows, such as in the lithium boiling tank, pressure can be strongly coupled with both velocity and density due to the significant variations in sound speed between phases. Therefore, a hybrid approach that incorporates both strong pressure-density and pressure-velocity coupling is required for accurate all-Mach compressible multiphase simulations. 

Before giving the formulation of the proposed hybrid approach for all-Mach flows, we first review one of the widely-used all-Mach Riemann solvers.

\subsection{Simple low-dissipation AUSM Riemann solver} \label{sec:slau}
One of the widely-used Riemann solvers is Advection upstream splitting method (AUSM) which was proposed in \cite{liou1993new} and is a robust and accurate method to treat both linear and nonlinear waves in complex flow \cite{niu2016computations,niu2019development}. The AUSM schemes split the inviscid numerical flux into convective and acoustic parts and adaptively reconstruct the flux according to the local Mach number. One sequel of AUSM-type Riemann is SLAU (Simple Low-dissipation AUSM) \cite{shima2011parameter,kitamura2013towards} scheme. The SLAU scheme is formulated to reduce the numerical dissipation in low Mach region and is free from reference parameters. First of all, Eq.~\ref{eq:invEquation} is solved by the finite volume method and written into the following semi-discrete form
\begin{equation}
   \Omega_i \frac{d \overline{\mathbf{U}}_i}{d t}+\sum_j \tilde{\mathbf{F}}_{inv,ij} \Gamma_{ij}=0,
\end{equation}
where $\Omega_i$ is the volume of cell $i$, $\overline{\mathbf{U}}_i$ are cell averaged solution variables, $\tilde{\mathbf{F}}_{inv,ij}$ is the discretized inviscid flux at cell face between cell $i$ and $j$, and $\Gamma_{ij}$ is the size of cell face between cell $i$ and $j$. According to the idea of SLAU, the inviscid numerical flux $\tilde{\mathbf{F}}_{inv}$ for the multi-component system can be split into 
\begin{eqnarray}
\tilde{\mathbf{F}}_{inv} & = & \frac{\dot m+|\dot m|}{2}\mathbf{\Psi}^{L}+\frac{\dot m-|\dot m|}{2}\mathbf{\Psi}^{R}+\mathcal{P}\mathbf{N}, \label{eq:slau_flux}
\end{eqnarray}
with vectors defined as
\begin{eqnarray*}
	\mathbf{\Psi}=\left[Y_1,...Y_k,\,\mathbf{u},\,H\right]^{\mathrm{T}} &, & \mathbf{N}=\left[0,...0,\,\mathbf{n},\,0\right]^{\mathrm{T}},
\end{eqnarray*}
where superscripts $L/R$ denote the left and right states of physical fields at the
cell-interface and $H$ stands for total mixture enthalpy $H=\left.\left(\rho E+p\right)\right/\rho$. $\mathbf{n}$ is the normal unit vector to the cell boundary. The mass flux $\dot m$ is 
\begin{eqnarray}
&  & \dot m=\frac{1}{2}\rho^{L}\left(V^{L}+\left|\overline{V}\right|^{L}\right)+\frac{1}{2}\rho^{R}\left(V^{R}-\left|\overline{V}\right|^{R}\right)-\frac{\chi}{\bar{c}}(p^{R}-p^{L}), \label{eq:slau_convective}
\end{eqnarray}
where $V={\bf u}\cdot \mathbf{n}$ stands for the velocity normal to cell boundary. Other quantities in  Eq.\eqref{eq:slau_convective} are computed as follows, 
\begin{eqnarray*}
	&  & \left|\overline{V}\right|^{L}=\left(1-\phi\right)\left|\overline{V}\right|+\phi\left|V^{L}\right|,\ \left|\overline{V}\right|^{R}=\left(1-\phi\right)\left|\overline{V}\right|+\phi\left|V^{R}\right|,\\
	&  & \left|\overline{V}\right|=\frac{\rho^{L}\left|V^{L}\right|+\rho^{R}\left|V^{R}\right|}{\rho^{L}+\rho^{R}},\ \chi=\left(1-\widehat{M}\right)^{2},\ \bar{c}=\frac{c^{L}+c^{R}}{2}\\
	&  & \widehat{M}=\min\left(1.0,\,\frac{1}{\bar{c}}\sqrt{\frac{|\mathbf{u}^{L}|^{2}+|\mathbf{u}^{R}|^{2}}{2}}\right),\ M^{L/R}=\frac{V^{L/R}}{\bar{c}},\\
	&  & \phi=-\max\left[\min\left(M^{L},0\right),-1\right]\cdot\min\left[\max\left(M^{R},0\right),-1\right].
\end{eqnarray*}
The pressure flux $\mathcal{P}$ is computed by the improved version proposed in \cite{kitamura2013towards} as 
\begin{eqnarray}
&  & \mathcal{P}=\frac{p^{L}+p^{R}}{2}+\frac{f_{p}^{L}-f_{p}^{R}}{2}\left(p^{L}-p^{R}\right)+\sqrt{\frac{|\mathbf{u}^{L}|^{2}+|\mathbf{u}^{R}|^{2}}{2}}\left(f_{p}^{L}+f_{p}^{R}-1\right)\bar{\rho}\bar{c}, \label{eq:slau_acoustic}
\end{eqnarray}
with
\begin{eqnarray*}
	&  & f_{p}^{L/R}=\left\{ \begin{array}{ll}
		\frac{1}{2}\left(1\pm\mathrm{sign}(M^{L/R})\right), & \mathrm{if}~\left|M^{L/R}\right|\geqslant1\\
		\frac{1}{4}\left(M^{L/R}\pm1\right)^{2}(2\mp M^{L/R}), & \mathrm{otherwise}
	\end{array}\right.,\ \bar{\rho}=\frac{\rho^{L}+\rho^{R}}{2}.
\end{eqnarray*}
The above improved method to calculate the pressure flux is denoted as SLAU2. These formulations show that the numerical flux calculated by SLAU is adaptive to the local Mach number. 

Although the AUSM family schemes such as the above SLAU2 improve the accuracy for the low Mach number flow, a preconditioned matrix is still required to simulate incompressible flows. As demonstrated by the work \cite{chen2022anti}, the Godunov-type scheme contains the deteriorated-accuracy pressure source term $p_d$ which is directly proportional to the density, sound speed and velocity difference:
\begin{equation}
    p_d \propto - \rho c \Delta U.
\end{equation}
Following the method proposed in \cite{chen2022anti}, the quantity of the SLAU2 scheme with the identical physical scale as $p_d$ is
\begin{equation}
    \sqrt{\frac{|\mathbf{u}^{L}|^{2}+|\mathbf{u}^{R}|^{2}}{2}}\left(f_{p}^{L}+f_{p}^{R}-1\right)\bar{\rho}\bar{c}\approx-\frac{3}{4}\sqrt{\frac{|\mathbf{u}^{L}|^{2}+|\mathbf{u}^{R}|^{2}}{2}}\frac{\rho^{L}+\rho^{R}}{2} (|\mathbf{U}^R|-|\mathbf{U}^L|) \bar{c} \propto - \rho c \Delta U.
\end{equation}
Thus, the SLAU2 scheme still suffers from the deteriorated accuracy source term as other Godunov-type schemes. To overcome the defect of pressure flux of the Godunov-type scheme, we propose a novel hybrid approach for all-Mach number flow simulations.

\subsection{Hybrid approach for all-Mach flows. First step: updating advection flux}
The proposed hybrid approach uses the Godunov-type SLAU2 scheme to calculate the advection flux. The advection flux $\tilde{\mathbf{A}}$ is calculated as
\begin{eqnarray}
\tilde{\mathbf{A}} & = & \frac{\dot m+|\dot m|}{2}\mathbf{\Psi}^{L}+\frac{\dot m-|\dot m|}{2}\mathbf{\Psi}^{R}. \label{eq:advection_flux}
\end{eqnarray}
For the multiphase system, $\dot m$ is modified as
\begin{eqnarray}
&  & \dot m=\frac{1}{2}\rho^{L}\left(V^{L}+\left|\overline{V}\right|^{L}\right)+\frac{1}{2}\rho^{R}\left(V^{R}-\left|\overline{V}\right|^{R}\right)-\theta\frac{\chi}{\bar{c}}(p^{R}-p^{L}), \label{eq:modifiedMassFlux}
\end{eqnarray}
where $\theta$ is the material interface indicator and is calculated as
\begin{eqnarray}
	&  & \theta=\left\{ \begin{array}{ll}
		0, & \mathrm{if}~0.01<\alpha_1<0.99~\mathrm{and}~|\mathbf{f_s}|>0,\\
		1, & \mathrm{otherwise}
	\end{array}\right..
\end{eqnarray}
Here, $\alpha_1$ is the volume fraction of the liquid phase. In Eq.~\ref{eq:modifiedMassFlux}, the last term $\frac{\chi}{\bar{c}}(p^{R}-p^{L})$ is the pressure-velocity coupling term for mass flux at low Mach number. When Ma$>1$, $\chi=0$ and the pressure-velocity coupling term for mass flux vanishes. When the surface tension effect exists, $\Delta p \neq 0$ across the material interface. To prevent spurious velocities caused by surface tension forces in the advection flux, this work modifies $\dot m$ such that the pressure-velocity coupling term vanishes at material interfaces where surface tension is present. 

With advection flux $\tilde{\mathbf{A}}$, the solution variable $\overline{\mathbf{U}}_i^n$ at time level $n$ are updated with the second-order SSP-RK (strong-stability-preserving Runge Kutta) time scheme \cite{gottlieb2009high} as
\begin{equation}
\begin{split}
   \overline{\mathbf{U}}_i^{\ast}&=\overline{\mathbf{U}}_i^{n}+\Delta t \mathbf{Res}(\overline{\mathbf{U}}_i^n), \\
   \overline{\mathbf{U}}_i^{\ast\ast}&=\frac{1}{2}\biggl(\overline{\mathbf{U}}_i^{n}+\overline{\mathbf{U}}_i^{\ast}+\Delta t \mathbf{Res}(\overline{\mathbf{U}}_i^{\ast})\biggl),
\end{split}  \label{eq:rk2} 
\end{equation}
where $\overline{\mathbf{U}}_i^{\ast\ast}$ are the updated intermediate solution using advection flux and the 2nd order SSP-RK scheme. $\mathbf{Res}(\overline{\mathbf{U}}_i)$ is the residual function and $\mathbf{Res}(\overline{\mathbf{U}}_i)=-\frac{\sum_j \tilde{\mathbf{A}}(\overline{\mathbf{U}}_i) \Gamma_{ij}}{\Omega_i}$ here. Through Eqs.~\ref{eq:rk2}, we obtain the intermediate solution as $(\rho Y_k)^{\ast\ast}$, $(\rho \mathbf{u})^{\ast\ast}$ and $(\rho E)^{\ast\ast}$. Since the pressure term does not affect the mass conservation equations, updating the advection flux leads to the solution at time level $n+1$ for mass as $(\rho Y_k)^{n+1}=(\rho Y_k)^{\ast \ast}$ and $\rho^{n+1}=\sum_{k=1}^N (\rho Y_k)^{n+1}$.

\subsection{Hybrid approach for all-Mach flows. Second step: updating advection flux}
In order to update the intermediate momentum and energy to time level $n+1$, we solve the following rest of the control equation system 
\begin{equation} \label{eq:interMomentum}
     \frac{(\rho \mathbf{u})^{n+1}-(\rho \mathbf{u})^{\ast \ast}}{\Delta t} =\mathbf{\nabla} p^{n+1},
\end{equation}
and
\begin{equation} \label{eq:interEnergy}
  \frac{(\rho E)^{n+1}-(\rho E)^{\ast \ast}}{\Delta t} =\mathbf{\nabla} \cdot (p \mathbf{u})^{n+1}.   
\end{equation}
As pointed out in Section \ref{sec:slau}, the pressure term in Godunov-type SLAU schemes deteriorates the accuracy for low-Mach number flow. Therefore, this work updates the pressure term by solving the pressure equation which is derived as follows. First, we adopt the standard incompressible flow formulation, which applies the momentum equation to derive the implicit equation for pressure. Thus, we divide Eq.~\ref{eq:interMomentum} and take its divergence to obtain:
\begin{equation} \label{eq:divergenceU}
    \nabla \cdot \mathbf{u}^{n+1}= \nabla \cdot \mathbf{u}^{\ast \ast}-\Delta t \nabla \cdot \biggl(\frac{\nabla p ^{n+1}}{\rho^{n+1}}\biggl), 
\end{equation}
where $\mathbf{u}^{\ast \ast}=(\rho \mathbf{u})^{\ast \ast}/\rho^{n+1}$ is the provisional velocity which excludes the pressure effect. Considering compressible flow and the homogeneous single-fluid model assumption, we apply the following evolution equation for mixture pressure which is derived in \cite{fedkiw2002general} as
\begin{equation} \label{eq:pressureEvolution}
    \frac{d p}{d t}+\mathbf{u} \cdot \nabla p=-\rho c^2 \nabla \cdot \mathbf{u}.
\end{equation}
By combining Eq.~\ref{eq:divergenceU} and Eq.~\ref{eq:pressureEvolution}, and assuming that $\nabla \cdot \mathbf{u}$ is fixed through the time step, we derive
\begin{equation}
    \frac{1}{\rho c^2}\frac{d p}{d t}+\frac{1}{\rho c^2} \mathbf{u} \cdot \nabla p=-\nabla \cdot \mathbf{u}^{\ast \ast}+\Delta t \nabla \cdot \biggl(\frac{\nabla p ^{n+1}}{\rho^{n+1}}\biggl).
\end{equation}
We approximate the above time derivative with the forward Euler time scheme. To simplify the above equation, we also apply the advection equation for pressure, which reads 
\begin{equation}\label{eq:advectionPressure}
    p^{\ast \ast}=p^{n}-(\mathbf{u}^n \cdot \nabla p^n) \Delta t.
\end{equation}
Finally, we can obtain the following pressure Helmholtz equation as 
\begin{equation}\label{eq:pressureEquation}
    \biggl(-\frac{1}{\rho^{n+1} (c^{\ast \ast})^2} \frac{1}{\Delta t}+\mathbf{\nabla} \cdot \Bigl( \frac{\Delta t}{\rho^{n+1}} \mathbf{\nabla} \Bigl) \biggl) p^{n+1}=-\frac{1}{\rho^{n+1}(c^{\ast \ast})^2} \frac{p^{\ast \ast}}{\Delta t} +\mathbf{\nabla} \cdot \mathbf{u}^{\ast \ast},
\end{equation}
where $p^{\ast \ast}$ is the intermediate pressure and $c^{\ast \ast}$ is the intermediate sound speed. It is noteworthy that in order to obtain the correct solution for high-Mach number flow, intermediate $p^{\ast \ast}$ should be calculated by using the conservative form, i.e., by applying EOS after updating advection flux in Eq.~\ref{eq:rk2}. The advection equation Eq.~\ref{eq:advectionPressure} is only used here for deriving the pressure Helmholtz equation. It is also noted here that after applying the finite volume method to the pressure Helmholtz equation, the last term in the right-hand side of Eq.~\ref{eq:pressureEquation}  becomes
\begin{equation}
    \int_{\Omega_i} \nabla\mathbf{u}^{\ast \ast}=\sum_j u^{\ast \ast}_{ij} \Gamma_{ij}.  
\end{equation}
In order to ensure that $u^{\ast \ast}_{ij}$ is calculated in the spirit of all-Mach number flows, upwinding should be applied for high-speed flows while velocity-pressure coupling should be considered for low-Mach number flows. Thus, we apply the SLAU Riemann solver here again to calculate discretized interface velocity $u^{\ast \ast}_{ij}$ as
\begin{equation} \label{eq:interfaceVelocity}
    u^{\ast \ast}_{ij}=\frac{\dot m^{\ast \ast}}{\rho_{ij}^{n+1}},
\end{equation}
where intermediate mass flux $\dot m^{\ast \ast}$ is obtained following Eq.~\ref{eq:modifiedMassFlux} and $\rho_{ij}^{n+1}$is approximated by $(\rho^{L}+\rho^{R})/2$ at time level $n+1$. The detailed algorithm for solving Eq.~\ref{eq:pressureEquation} using the finite volume method will be presented in the following section. After solving $p^{n+1}$ from Eq.~\ref{eq:pressureEquation}, the momentum $(\rho \mathbf{u})^{n+1}$ and energy $(\rho E)^{n+1}$ at time level $n+1$ can be obtained by using Eq.~ \ref{eq:interMomentum} and Eq.~\ref{eq:interEnergy}.  

\subsection{Properties and remarks of the proposed hybrid approach}
In this subsection, we examine the properties of the proposed hybrid approach. Some remarks are also given to clarify the proposed method. 
\begin{proposition*}
The proposed hybrid method is capable of solving the stationary contact discontinuity exactly.     
\end{proposition*}
\begin{proof}
Without compromising the generality, we assume that the flow system comprises two components and the stationary contact discontinuity is located between cell $i$ and $j$. Thus, there is $u_i^n=u_j^n=0$, $p_i^n=p_j^n$ and $(\rho Y_1)_i^n \neq (\rho Y_1)_j^n$. Following the first step of the proposed method, there is $\dot m=0$ from Eq.~\ref{eq:modifiedMassFlux} and $\tilde{\mathbf{A}}=\mathbf{0}$ from Eq.~\ref{eq:advection_flux}. Thus, from Eq.~\ref{eq:rk2}, we can conclude that the intermediate state is the same as the initial condition. Therefore, the intermediate pressure $p^{\ast \ast}=p^n$. 

To calculate the $p^{n+1}$ from the second step. We write the pressure at time level $n+1$ as $p^{n+1}=p^{n}+\delta p$. Then, the pressure Helmholtz equation Eq.~\ref{eq:pressureEquation} becomes
\begin{equation*}
        \biggl(-\frac{1}{\rho^{n+1} (c^{\ast \ast})^2} \frac{1}{\Delta t}+\mathbf{\nabla} \cdot \Bigl( \frac{\Delta t}{\rho^{n+1}} \mathbf{\nabla} \Bigl) \biggl) \biggl(p^{n}+\delta p \biggl)=-\frac{1}{\rho^{n+1}(c^{\ast \ast})^2} \frac{p^{n}}{\Delta t} +\mathbf{\nabla} \cdot \mathbf{u}^{\ast \ast}
\end{equation*}
From the definition of contact discontinuities, there is $\nabla p^n=\mathbf{0}$. We can also conclude that $\nabla \cdot \mathbf{u}^{\ast \ast}=0$ from Eq.~\ref{eq:interfaceVelocity}. Thus, the above pressure Helmholtz equation is simplified as
\begin{equation*}
    \delta p=\rho^{n+1} (c^{\ast \ast})^2 \Delta t^2 \mathbf{\nabla} \cdot \Bigl( \frac{\mathbf{\nabla} \delta p}{\rho^{n+1}} \Bigl).
\end{equation*}
By solving the above equation, we can obtain $\delta p=0$ and thus $\nabla p^{n+1}=\nabla p^{n}=\mathbf{0}$. Combining the first step of the solution procedure where $\tilde{\mathbf{A}}=\mathbf{0}$ and the second step where $\nabla p^{n+1}=\nabla p^{n}=\mathbf{0}$, we can conclude the stationary contact discontinuity is exactly solved by the proposed hybrid method.
\end{proof}
Since material interfaces can be considered as contact discontinuities when excluding the surface tension effect, the capability of solving stationary contact exactly is essential for multiphase flow simulations. To further highlight the properties of the proposed method, we give the following remarks.
\begin{description}
\item[Remark 1.]{The pressure Helmholtz equation eliminates the accuracy-deteriorated pressure term in the Godunov scheme and establishes the coupling between pressure-velocity and pressure-density. Thus it can be applied to all Mach number flow regimes. For supersonic flows, the upwind scheme is employed by the mass flux of the proposed hybrid method, which retains the original property of the SLAU scheme. For incompressible limits, the central scheme is recovered by the mass flux and the pressure Helmholtz equation becomes
\begin{equation}
    \mathbf{\nabla} \cdot \Bigl( \frac{\Delta t}{\rho^{n+1}} \mathbf{\nabla} p^{n+1}\Bigl) =\mathbf{\nabla} \cdot \mathbf{u}^{\ast \ast},
\end{equation}
which is similar to the standard projection solution procedure and satisfies the divergence-free condition. 
}
\item[Remark 2.]Unlike the approach in \cite{demou2022pressure}, the proposed solution procedure solves the conservative form, making it well-suited for flows involving shock waves. However, as the Godunov-type scheme is applied in the first step and the intermediate pressure is solved using EOS, the proposed approach requires special techniques to prevent numerical oscillations across material interfaces. The mechanism of generating oscillations by conservative schemes has been well discussed in \cite{abgrall2001computations}. 
\end{description}

\section{Modelling of multi-physical processes } \label{sec:multiphysics}
Besides the convection term which is the essential part of all-Mach number flow modelling, additional physical processes including viscosity, heat conduction, gravity force, surface tension and phase transition should also be modelled. The modelling strategies for these additional physical processes are introduced in this section. 

\subsection{Surface tension model} \label{sec:surfaceTension}
We follow the continuum surface force (CSF) model proposed by \cite{brackbill1992continuum} to model the surface tension as
\begin{equation}
    \mathbf{f_s}=\sigma \kappa \nabla \alpha,
\end{equation}
where $\sigma$ is the surface tension coefficient and $\kappa$ is the curvature calculated as
\begin{equation}
    \kappa= - \nabla \cdot \mathbf{\hat{n}}=-\nabla \cdot \Bigl( \frac{\nabla \alpha}{|\nabla \alpha|} \Bigl),
\end{equation}
where $\mathbf{\hat{n}}$ is the unit normal vector at the interface. In order to calculate the interface normal vectors more accurately, this work adopts the smoothed CSF model proposed in \cite{raeini2012modelling}. The smoothed volume fraction field used for calculating the interface normal vectors is determined by interpolating it between cell centers and face centers recursively as 
\begin{equation}
    \alpha_{s,l+1}=0.5 \langle \langle \alpha_{s,l} \rangle_{c\rightarrow f} \rangle_{f\rightarrow c}+0.5 \alpha_{s,l}, ~~\alpha_{s,0}=\alpha,
\end{equation}
where $\alpha_{s,l}$ represents the smoothed volume fraction field after $l$th iteration. The operators $\langle \rangle_{c\rightarrow f}$ and  $\langle \rangle_{f\rightarrow c}$ represent interpolation processes, where $\langle \rangle_{c\rightarrow f}$ interpolates a field from cell centers to face centers, and $\langle \rangle_{f\rightarrow c}$ interpolates a field from face centers to cell centers. In our simulation, we set $l=2$. After obtaining smoothed volume fraction $\alpha_{s,l+1}$, the unit normal vector at the interface is calculated as
\begin{equation}
    \mathbf{\hat{n}}=\frac{\nabla \alpha_{s,l+1}}{|\nabla \alpha_{s,l+1}|}. 
\end{equation}
It is noteworthy that other smoothing techniques such as the Gaussian filter can also be employed to obtain the smoothed volume fraction field.

\subsection{Phase transition model} \label{sec:phaseTransition}
We only consider the phase transition between the liquid denoted by phase 1 and its vapour denoted by phase 2 here. The source terms modelling the phase transition can be written as 
\begin{equation}
\begin{cases}
  \begin{aligned}
	&\dot\omega_{1} =  \rho \tau (g_2 -g_1), \\
	&\dot\omega_{2} = - \rho \tau (g_2 -g_1), 
  \end{aligned}
  \end{cases}
   \label{eq:PDE_phase_transition}
\end{equation}
where $g_k$ represents Gibbs free energy of phase $k$ and $\tau$ is the rate at which thermodynamic equilibrium is reached and is usually a function of the specific interfacial area, temperature and pressure. Instead of solving Eq.~\eqref{eq:PDE_phase_transition}, the present work employs a fast phase transition relaxation model proposed in \cite{chiapolino2017simple,chiapolino2017simple2} which assumes $\tau$ is very large and relaxation to thermodynamic equilibrium is immediate. 

During the phase transition, the mass fractions of the liquid and vapor phases vary, while the pressure and temperature adjust to their equilibrium values. However, the mixture specific volume $v$, mixture energy $e$ and mass fractions for non-condensable gases $Y_{k\geq3}$ remain constant. In the fast phase transition relaxation process, the fractional step method is used to achieve the equilibrium state $(Y_1^{\lozenge},Y_2^{\lozenge},p^{\lozenge},T^{\lozenge})$ where the vapour partial pressure $p_{partial}$ is equal to the saturation pressure $p_{sat}$. Thus, if the phase transition occurs, the following system is solved
\begin{equation}
\begin{cases}
p_{partial}=x^{\lozenge}_vp^{\lozenge}=p_{sat}(T^{\lozenge}), \\
v = Y^{\lozenge}_{1} v_1(T^{\lozenge}, p^{\lozenge}) + Y^{\lozenge}_{2} v_2(T^{\lozenge}, p^{\lozenge}) + \sum^{N}_{k=3} Y_{k} v_k(T^{\lozenge}, p^{\lozenge}),\\
e = Y^{\lozenge}_{1} e_1(T^{\lozenge}, p^{\lozenge}) + Y^{\lozenge}_{2} e_2(T^{\lozenge}, p^{\lozenge}) + \sum^{N}_{k=3} Y_{k} e_k(T^{\lozenge}, p^{\lozenge}),\\
\end{cases}
\label{eq:Phase_transition_thermoClosure_system}
\end{equation}
where vapour molar fraction $x^{\lozenge}_v$ is defined as
\begin{equation}
x^{\lozenge}_v=\frac{Y^{\lozenge}_{2}/{W}_{2}}{Y^{\lozenge}_{2}/{W_{2}}+\sum^{N}_{k=3}{Y_{k}/{W_{k}}}}.
\label{eq:eq_titrevap}
\end{equation}
Here, $W_k$ is the molecular weight of gas $k$.
Then the solution $Y_2^{\lozenge}$ can be approximated by solving the constraint conditions in Eq.~\ref{eq:Phase_transition_thermoClosure_system} respectively as
 \begin{equation}
 \begin{cases}
Y^{sat}_{2}(p,T)=\frac{p_{sat}(T) W_2}{p-p_{sat}(T)}\sum_{k\geq 3}Y_k/W_k, \\
Y^m_2(p,T) = \frac{v-v_g(p,T)}{v_2(p,T)-v_1(p,T)}, \\
Y^e_2(p,T) = \frac{e-e_g(p,T)}{e_2(p,T)-e_1(p,T)},
\end{cases}
 \label{eq:y2appro}
\end{equation}
where 
\begin{equation*}
v_g(p,T) = \left(1- \sum^N_{k=3} Y_k \right) v_2(p,T)  + \sum^N_{k=3} Y_k v_k(p,T),~~\text{and}~~e_g(p,T) = \left(1- \sum^N_{k=3} Y_k \right) e_2(p,T)  + \sum^N_{k=3} Y_k e_k(p,T). 
\end{equation*}
The solution of Eq.~\ref{eq:y2appro} is to find conditions satisfying  $Y^m_2(p^{\lozenge},T^{\lozenge})=Y^e_2(p^{\lozenge},T^{\lozenge})=Y^{sat}_{2}(p^{\lozenge},T^{\lozenge})$ which is not trivial to solve. The relaxation solver approximates the solution in the following way
\begin{itemize}
\item $Y^m =Y^m_2(p,T)$ and $Y^e=Y^e_2(p,T)$ are evaluated with the initial values of $(p, x_v)$ and $T=T_{sat}(x_vp)$ which are before phase transition relaxation process, 
\item $Y^{sat}=Y^{sat}_{2}(p,T)$ is evaluated at the initial $(p,T)$.
\end{itemize}
Then the variations of mass fraction during the phase transition are calculated as 
\begin{equation}
\begin{cases}
z_1 = \left(Y^{m}-Y_2\right) \left(Y^{e}-Y_2\right),\\
z_2 = \left(Y^{m}-Y_2\right) \left(Y^{sat}-Y_2\right),
\end{cases}
\label{eq_prodY}
\end{equation}
where $Y_2$ is the initial mass fraction before the relaxation phase transition. The idea of the selection process is to select the value with the smallest variation. The selection process states:
\begin{itemize}
\item If $z_1<0$, or $z_2<0$, no mass transfer happens:  $Y_2^{\lozenge}=Y_2$.
\item Otherwise, the one which produces the minimum variation of mass transfer is selected.
\end{itemize}
Under the second condition, $Y^{\lozenge}_2$ is calculated as:
\begin{equation}
Y^{\lozenge}_2 =  Y_2 + \text{sgn}\left[ Y^{m}-Y_2\right] \times \text{Min} \left[ | Y^m-Y_2 |, | Y^{e}-Y_2 |,  | Y^{sat}-Y_2 |  \right].
\label{eq_choixY}
\end{equation}
After $Y^{\lozenge}_2$ is determined, the other variables $Y^{\lozenge}_1, p^{\lozenge}, T^{\lozenge}$ at phase equilibrium state can be updated.  It is noteworthy that in this work the saturation property $p_{sat}(T)$ is obtained through a fitted Antoine equation 
 \begin{equation}
    p_{sat}(T)=10^{A-\frac{B}{C+T}},
\label{eq_Antoine}
\end{equation}   
with the $A,B,C$ parameters reported in the NIST database. As proposed by \citet{boivin2019thermodynamic}, this approach eliminates the iterative procedure required by the original NASG formulation, reducing numerical cost while maintaining accuracy.

\subsection{Overall algorithm for modelling multi-component flows at all-Mach number}
The viscosity, heat conduction, gravity and surface tension terms in Eq.~\ref{eq_PDE} are updated with the advection flux during the first step of the solution procedure and are added to the residual function $\mathbf{Res}$ in Eq.~\ref{eq:rk2}. In this work, we add the phase transition model after the second step of the solution procedure. The overall algorithm for modelling multi-physical processes is listed in Algorithm \ref{algorithm:overall}. It is noted that we implement the proposed algorithm for all-Mach number flows into unstructured-grid-based OpenFoam which can solve flows with irregular geometry. 

\begin{algorithm}
\caption{Overall algorithm for modelling multi-component flows at all-Mach number}\label{alg:cap}
\begin{algorithmic}[1]
\State Initialise primitive variables $Y_k$, $T$, $p$, $\mathbf{u}$. Calculate conservative variables $\mathbf{U}=\Bigl(\rho Y_k, \rho \mathbf{u}, \rho E\Bigl)^{\text{T}}$.
\State set $n=0$.
\While{$t <$ end time}
\State Set $n=n+1$. Set Runge Kutta sub-step $rk=0$.
    \While{$rk<2$}
    \State Calculate advection flux $\tilde{\mathbf{A}}$ using Eq.~\ref{eq:advection_flux} and Eq.~\ref{eq:modifiedMassFlux}, add to the residual function $\mathbf{Res}$.
    \State Calculate the viscosity, heat conduction, and gravity force, add to $\mathbf{Res}$.
    \State Calculate the surface tension by following Section~\ref{sec:surfaceTension}, add to $\mathbf{Res}$
    \State Update conservative variables $\mathbf{U}$ by Eq.~\ref{eq:rk2}.
    \EndWhile
\State Obtain $(\rho Y_k)^{n+1}$, $\rho^{n+1}$, and intermediate state $(\rho \mathbf{u})^{\ast \ast}$ and $(\rho E)^{\ast \ast}$.
\State Calculate the intermediate pressure $p^{\ast \ast}$ and sound speed $c^{\ast \ast}$.
\State Calculate $\nabla \cdot \mathbf{u}^{\ast \ast}$ using Eq.~\ref{eq:interfaceVelocity}.
\State Solve the pressure Helmholtz equation Eq.~\ref{eq:pressureEquation} and obtain pressure $p^{n+1}$.
\State Obtain $(\rho \mathbf{u})^{n+1}$ and $(\rho E)^{n+1}$ using Eq.~\ref{eq:interMomentum} and Eq.~\ref{eq:interEnergy}.
\If{phase transition}
    \State Calculate the phase equilibrium state $(Y^{\lozenge},p^{\lozenge},T^{\lozenge})$ by following Section~\ref{sec:phaseTransition}.
    \State Set $(Y^{n+1},p^{n+1},T^{n+1})$ using equilibrium state.
\EndIf
\EndWhile
\State End of simulation.
\end{algorithmic}
\label{algorithm:overall}  
\end{algorithm}

\section{Numerical methods}\label{sec:numericalMethods}
\subsection{Low-dissipation, scalar-structure-preserving ROUND schemes}
Reconstruction schemes are required to obtain the left and right state at cell boundaries in Eq.~\ref{eq:advection_flux}. In order to achieve a balanced accuracy for both discontinuous solutions and smooth solutions, this work employs the low-dissipation, scalar-structure-preserving ROUND schemes proposed in \cite{deng2023unified,deng2023new}. We adopted the ROUND schemes of the unstructured-grid-based version, of which source code has been released at DOI:10.5281/zenodo.8399000. In our solver, the ROUND schemes are applied to primitive variable $\phi$. To obtain the reconstructed variable $\phi_{ij}^{L}$ at the cell boundary for the left-side state, the reconstruction scheme reads
\begin{equation}
    \phi_{ij}^{L}=\phi_{i}+\frac{|\mathbf{d}_{i,ij}|}{|\mathbf{d}_{i,j}|} \psi(r_{ij}^L) \Bigl(\phi_j-\phi_i \Bigl),
\end{equation}
where $\mathbf{d}_{i,j}$ is the vector pointed from the mass center of cell $i$ to the mass center of cell $j$, and $\mathbf{d}_{i,ij}$ is the vector from cell $i$ to cell face $ij$. $\psi(r_{ij}^L)$ is the limited function which is calculated by the ROUND scheme. To simultaneously satisfy the low-dissipation, scalar-structure-preserving, and total variation diminishing (TVD) properties, a ROUND scheme is 
\begin{equation}
    \psi(r_{ij}^L) =\left\{
\begin{array}{ll}
\text{min}\biggl(\frac{(1+r_{ij}^L)^{8}(2-5r_{ij}^L)+6r_{ij}^L\bigl(\beta_1(0.5-0.5r_{ij}^L)^{4}+(1+r_{ij}^L)^{4}\Bigl)^2}{3\bigl(\beta_1(0.5-0.5r_{ij}^L)^{4}+(1+r_{ij}^L)^{4}\bigl)^2},2r_{ij}^L\biggl) & 0 < r_{ij}^L \leq 1 \\
\text{min}\biggl(\frac{6\beta_2(1-\lambda)(-1+r_{ij}^L)^{4}\bigl(\beta_2(1-r_{ij}^L)^{4}+32(1+r_{ij}^L)^{4}\bigl)+256(1+r_{ij}^L)^{8}(2+r_{ij}^L)}{3\bigl(\beta_2(1-r_{ij}^L)^{4}+16(1+r_{ij}^L)^{4}\bigl)^2},1.7\biggl) & 1 < r_{ij}^L \\
0 & r_{ij}^L \leq 0
\end{array}
\right.,
\end{equation}
where parameters $\beta_1=1100$ and $\beta_2=800$. The above ROUND scheme is denoted as ROUND\_A+ in \cite{deng2023new}. $r_{ij}^L$ is the left-side gradient ration for cell interface $ij$. Following the work \cite{jasak1996error}, $r_{ij}^{L}$ is calculated as
\begin{equation}
    r_{ij}^{L}=\frac{2 (\nabla \phi)_i \cdot \textbf{d}_{i,j}}{\phi_j-\phi_i}-1.
\end{equation} 
The reconstructed value $\phi_{ij}^{R}$ at the cell boundary for the right-side state can be obtained in a symmetrical manner. As demonstrated in the numerical tests, the ROUND\_A+ scheme produces less diffusive results across material interfaces compared to classic TVD schemes. More importantly, the ROUND scheme effectively preserves the structure of material interfaces, whereas classic TVD schemes usually distort the shape of interfaces.

\subsection{Homogeneous reconstruction strategy}
As noted in \cite{abgrall2001computations}, employing Godunov-type conservative schemes to solve the homogeneous four-equation model can result in pressure and velocity oscillations at the material interface. Therefore, special techniques, such as the non-conservative double-flux strategy \cite{jasak1996error}, should be implemented to prevent numerical oscillations. The alternative strategy proposed by \cite{williams2019fully,deng2020diffuse} is to choose the primitive variables that are thermodynamically consistent for the reconstruction process. For example, the set of primitive variables $\Bigl(\rho_i, \mathbf{u}, p, \alpha_i \Bigl)$ or $\Bigl(s_i, \mathbf{u}, p, \alpha_i \Bigl)$ is used in \cite{williams2019fully}, where $s_i$ is the specific entropy. In this work, given that the underlying assumption of a homogeneous model is that the material interfaces are at an isothermal and isobaric state, we will simply choose $\Bigl(T, \mathbf{u}, p, Y_i \Bigl)$ as the set of primitive variables by following the work \cite{deng2020diffuse}. Instead of using $\Bigl(\rho, \mathbf{u}, p, Y_i \Bigl)$, employing $\Bigl(T, \mathbf{u}, p, Y_i \Bigl)$ as the set of primitive variables for the reconstruction process is in line with the assumption of the homogenous model and can effectively prevent pressure and velocity oscillations. However, it should be noted that such a strategy cannot eliminate oscillations for the compressible multiphase flows where the non-equilibrium effect is strong. For the flow problems considered in this study, particularly those involving phase transitions where material interface cells are assumed to achieve phase equilibrium instantaneously, the homogeneous reconstruction strategy adopted here is well-suited. 

\subsection{Discretization of pressure equation}
We apply the finite volume method to solve the pressure equation in Eq.~\ref{eq:pressureEquation} as
\begin{equation}\label{eq:FVMpressureEquation}
    -\frac{p_i^{n+1}}{\rho_i^{n+1} (c^{\ast \ast})^2 \Delta t} + \frac{1}{\Omega_i} \sum_j\biggl(\Bigl( \frac{\Delta t}{\rho^{n+1}} \mathbf{\nabla} p^{n+1} \Bigl)_{ij} \cdot \mathbf{n}_{ij} \Gamma_{ij} \biggl) =-\frac{p^{\ast \ast}_i}{\rho_i^{n+1}(c^{\ast \ast})^2} +\frac{1}{\Omega_i} \sum_j u_{ij}^{\ast \ast} \Gamma_{ij}.
\end{equation}
This work discretizes the elliptic flux of the pressure variation by following the deferred-correction method proposed in \cite{xie2016multi} as 
\begin{equation}\label{eq:deferred}
    (\nabla p)_{ij} \cdot \mathbf{n}_{ij}=\frac{p_j^{l+1}-p_i^{l+1}}{|\mathbf{d}_{i,j}|}+\Bigl( (\nabla p)^{l}_{ij}-\frac{p_j^l-p_i^l}{|\mathbf{d}_{i,j}|}\mathbf{n}_{ij} \Bigl)\cdot \mathbf{n}_{ij},
\end{equation}
where superscript $l$ represents the iteration step. Combining Eq.~\ref{eq:FVMpressureEquation} and Eq.~\ref{eq:deferred}, the semi-implicit formulation results in a set of linear equation as
\begin{equation}
    -\frac{p_i^{l+1}}{\rho_{ij}^{n+1} (c^{\ast \ast})^2\Delta t}+\sum_j \frac{\Gamma_{ij} \Delta t}{\rho_{ij}^{n+1} |\mathbf{d}|_{i,j}}(p_j^{l+1}-p_i^{j+1})=-\frac{p_i^{\ast\ast}}{\rho_{ij}^{n+1} (c^{\ast \ast})^2\Delta t}+\sum_j\biggl(\Gamma_{ij}\Bigl(u_{ij}^{\ast \ast}-\frac{\Delta t}{\rho_{ij}^{n+1}} \mathbf{n}_{ij}\bigl((\nabla p)^{l}_{ij}-\frac{p_j^l-p_i^l}{|\mathbf{d}_{i,j}|}\mathbf{n}_{ij}\bigl) \Bigl) \biggl).
\end{equation}
The above linear equation system is solved iteratively until the residual is diminished below a give tolerance. Once the updated pressure field $p^{n+1}$, the pressure gradient $\nabla p^{n+1}$ can be calculated by the least-square approach following the work \cite{mahesh2004numerical}.

\section{Numerical experiments}\label{sec:numericalExperiments}
In this section, we conduct numerical experiments involving high-speed compressible and low-speed incompressible multi-component flows to validate the proposed solver's capability to handle a wide range of Mach numbers and to model multi-physics phenomena including phase transition. The thermodynamic parameters for the materials considered in the numerical experiments are provided in Table~\ref{tab:waterThermal}.   

\begin{table}[]
\centering
\caption{Thermodynamic parameters for materials used in the present numerical experiments}
\label{tab:waterThermal}
\begin{tabular}{lllllllll}
\hline
Phase & $C_v$ (J/kg/K) & $\gamma$ & $P_{\infty}$ (Pa) & $q$ (J/kg)  & $b$ ($\text{m}^3/kg$)     \\ \hline
water liquid & 3610 & 1.19 & $7.02\times10^{8}$ & -1177788 & $6.61\time 10^{-4}$   \\
water vapour & 955 & 1.47 & 0 & 2077616 & 0   \\
aerogel & 8 & 1.4 & 0 & 0 & 0 \\
air & 719 & 1.4 & 0 & 0 & 0    \\ 
helium &3110 &1.4 & 0 & 0 & 0    \\ \hline
\end{tabular}
\end{table}

\subsection{Rising bubble problems}
In this subsection, we will simulate the rising bubble problem which is a well-known benchmark test for incompressible two-phase flows \cite{hysing2009quantitative}. This problem has been extensively studied by \citet{gamet2020validation}. We used this test to assess the ability of the proposed method to handle incompressible two-phase flows and accurately capture interface topological changes in the presence of surface tension. Two test cases from \citet{hysing2009quantitative} are considered. The first involves moderate bubble deformation with density ratio $\rho_l/\rho_g \approx 10$, while the second features highly complex topological changes due to a large density ratio ($\rho_l/\rho_g \approx 1000$) and viscosity ratio.

The computational domain is $[0,1]\times[0,2]$ $\text{m}^2$. A static circular bubble with a diameter of 0.5 m is initially positioned at $(0.5, 0.5)$ in the static fluid. The initial conditions and physical parameters for both cases are provided in Table \ref{tab:riseBubblePhy}. The thermodynamic parameters for liquid water and aerogel are used in Case 1, while those for liquid water and air are applied to Case 2. Using these thermodynamic parameters, the initial density ratio is approximately $\rho_l/\rho_g \approx 10$ for Case 1 and $\rho_l/\rho_g \approx 1000$ for Case 2. By introducing the initial bubble radius $r_0$ as the reference length, we can define the dimensionless Reynolds number as $Re=\rho_l r_0 \sqrt{2|\mathbf{g}|r_0}/\mu_l$ and the E\"otv\"os number as $Eo=4\rho_l|\mathbf{g}|r_0^2/\sigma$. Then for Case 1, $Re=420$ and $Eo=10$. For Case 2, $Re=300$ and $Eo=125$.  The boundary conditions on the side walls are free-slip, while the top and bottom walls have no-slip boundary conditions. Hexahedral meshes of size $h=6.25\times 10^{-3}$ are used for the simulation. The simulation is run until $t=3.0$s. 

The interface topological changes for Case 1 at different time instants are presented in Fig.~\ref{fig:interfaceRise1}. Case 1 is characterized by a high surface tension coefficient and low density and viscosity ratios. As a result, the bubble undergoes moderate deformation, primarily driven by surface tension effects. The initial circular bubble becomes an ellipsoidal shape at $t=3.0$s, as shown in Fig.~\ref{fig:interfaceRise1}(d). The reference solutions for the shape of the interface, taken from \citet{hysing2009quantitative}, are included in Fig.~\ref{fig:interfaceRise1}(d) and marked with circle symbols. It can be seen that the current solution agrees well with the reference solution. For quantitative comparisons, the location of the bubble's center of mass and its rising velocity are presented by the solid black lines in Fig.~\ref{fig:risingBubbleFirstPlot} where the reference solutions from \cite{hysing2009quantitative} and numerical solutions from \cite{manzanero2020entropy} are also included. It can be seen that the current solutions agree well with the reference solution of \cite{hysing2009quantitative}. In contrast, the numerical solution using the entropy-stable artificial compressibility method exhibits oscillatory behaviours, which may be caused by their boundary condition implementation. Similar oscillations are also observed in the preconditioning method used in \cite{yoo2021homogeneous}. However, our proposed method produces significantly fewer oscillations, demonstrating the accuracy and robustness of the current method to solve multiphase flows in the incompressible limit.

\begin{table}[]
\centering
\caption{Initial conditions and physical parameters for the rising bubble problems}
\label{tab:riseBubblePhy}
\begin{tabular}{lllllllll}
\hline
Case & $p_0$ (Pa) & $T_0$ (K) & $\mu_l$ (Pa $\cdot$ s) & $\mu_g$ (Pa $\cdot$ s)  & $\sigma$ (N/m) & $|\mathbf{g}|~(\text{m}/s^2)$    \\ \hline
1 & 101000 & 298 & 10 & 1 & 24.5 &  0.98  \\
2 & 101000 & 298 & 10 & 0.1 & 1.96 & 0.98   \\ \hline
\end{tabular}
\end{table}

\begin{figure}
	\begin{center}
	\subfigure[t=0.5 s]{\includegraphics[width=.24\textwidth,trim={11.5cm 2.5cm 11.5cm 1.5cm},clip]{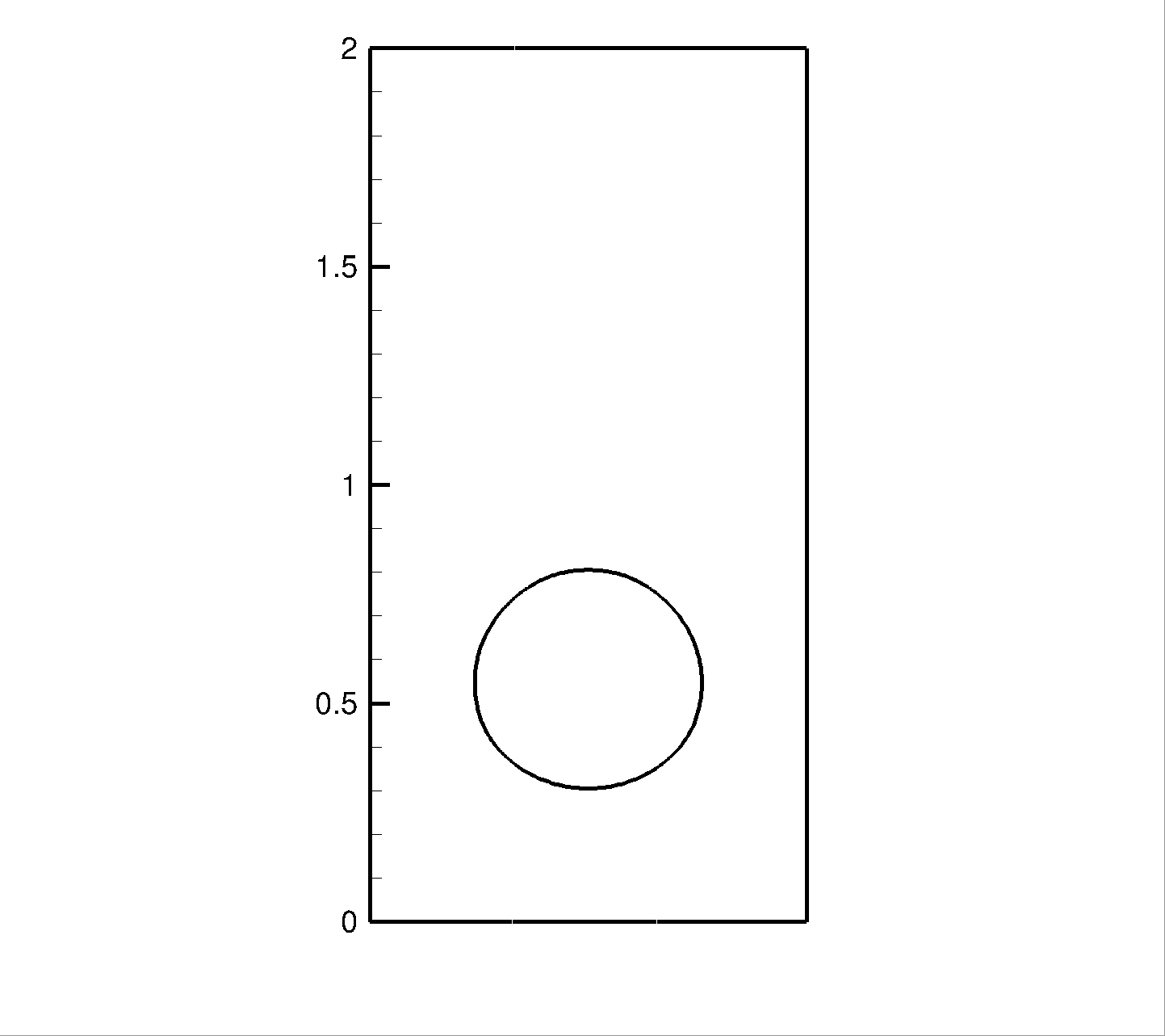}}
	\subfigure[t=1.0 s]{\includegraphics[width=.24\textwidth,trim={11.5cm 2.5cm 11.5cm 1.5cm},clip]{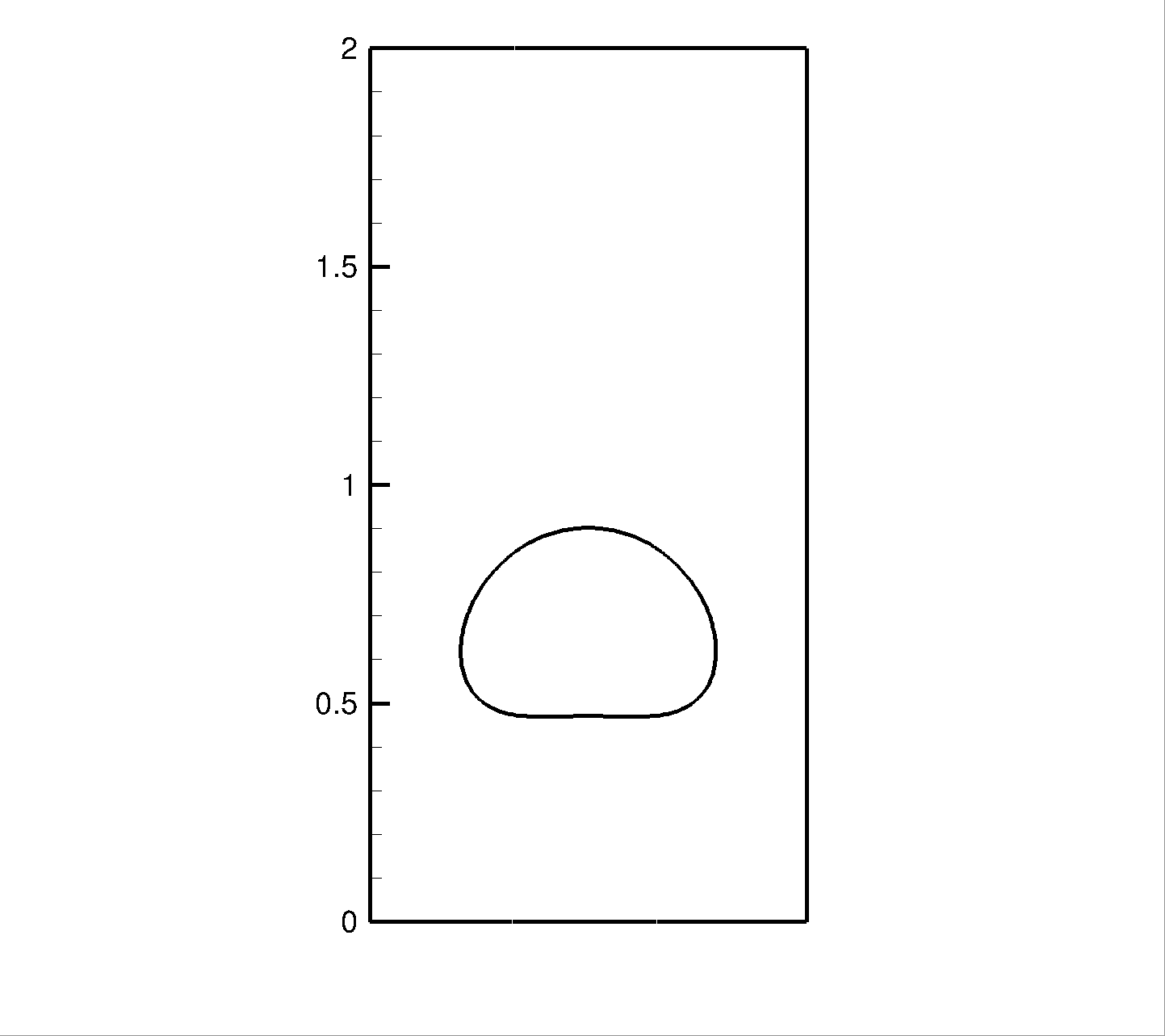}}
 	\subfigure[t=2.0 s]{\includegraphics[width=.24\textwidth,trim={11.5cm 2.5cm 11.5cm 1.5cm},clip]{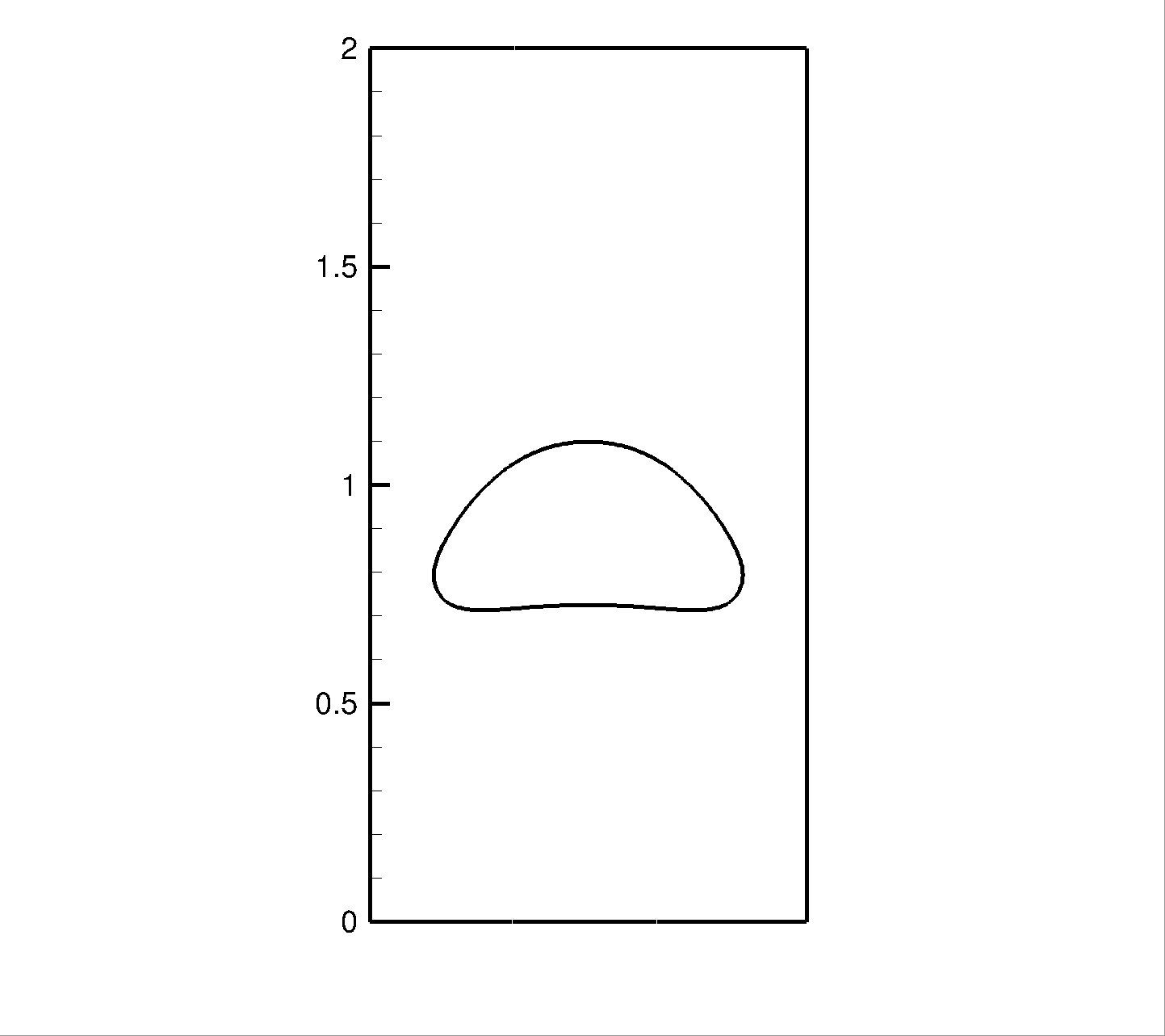}}
	\subfigure[t=3.0 s]{\includegraphics[width=.24\textwidth,trim={11.5cm 2.5cm 11.5cm 1.5cm},clip]{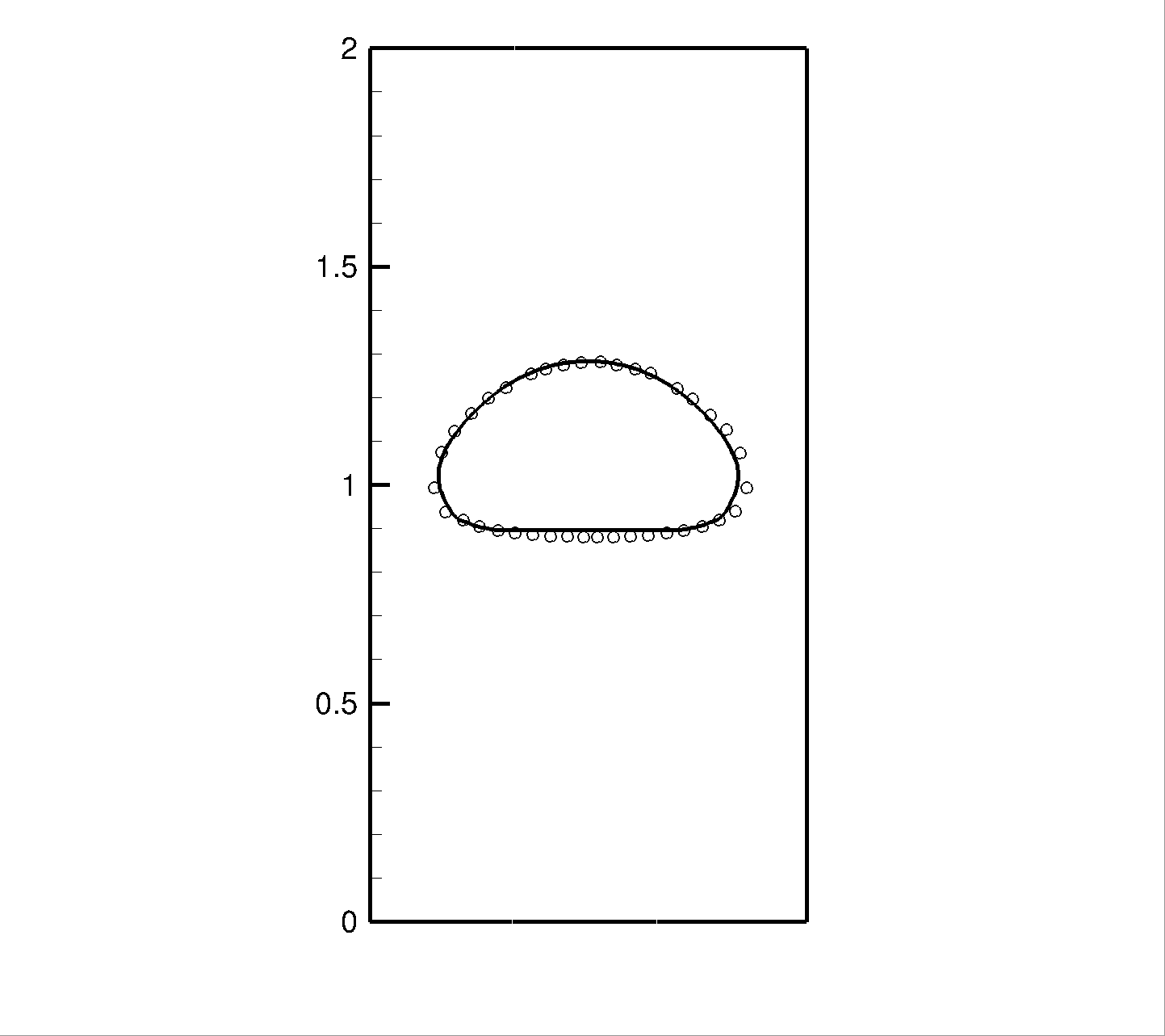}}
	\end{center}
	\protect\caption{Interface topological changes at different time instants for the case 1 rising bubble problem. The solid black lines show the interface. The circle symbols at the final time instant are the reference solution from \cite{hysing2009quantitative}.
    \label{fig:interfaceRise1}}	
\end{figure}

\begin{figure}
	\begin{center}
	\subfigure[position of mass center]{\includegraphics[width=.42\textwidth,trim={0.5cm 0.5cm 0.5cm 0.5cm},clip]{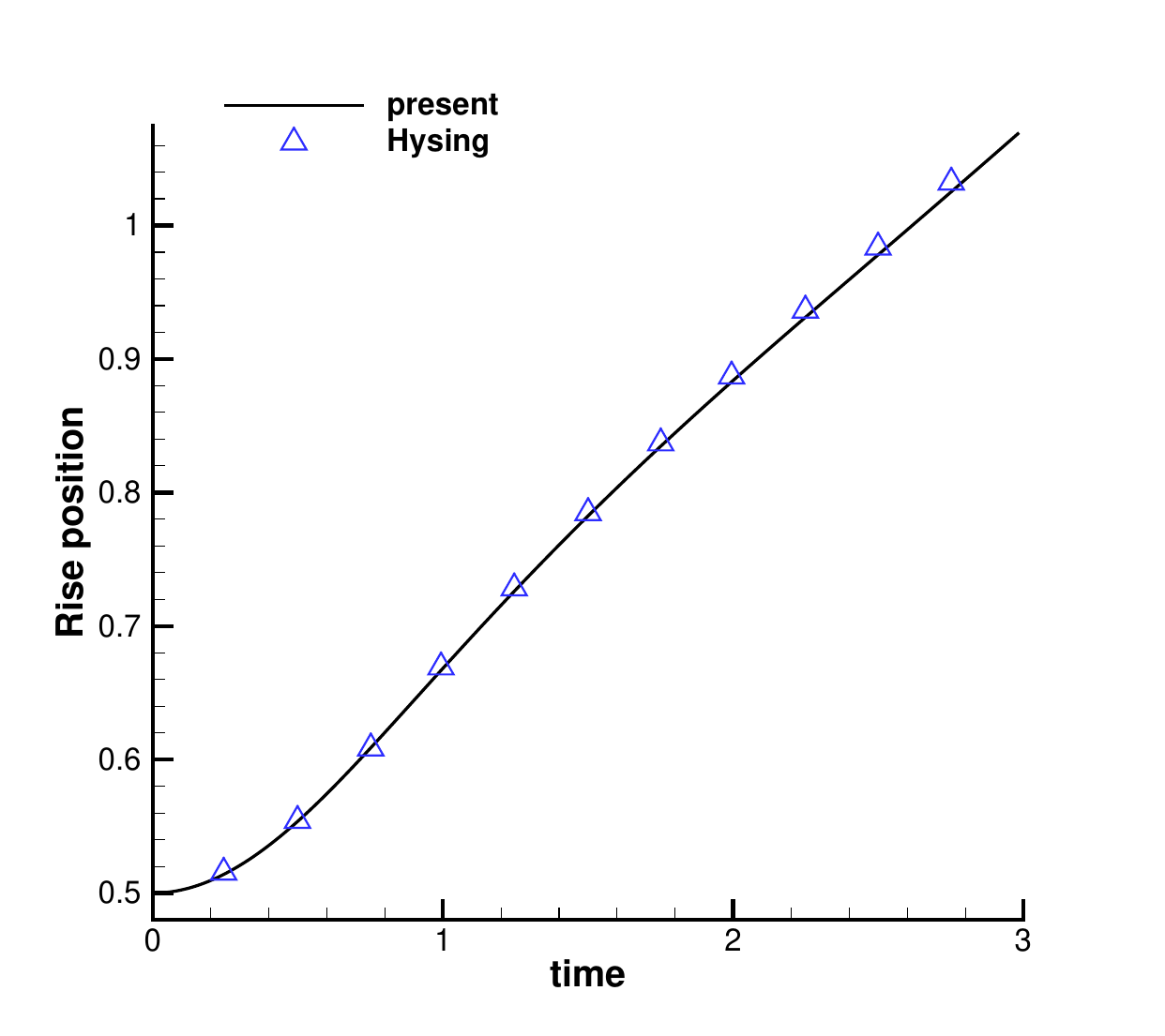}}
	\subfigure[rising velocity]{\includegraphics[width=.42\textwidth,trim={0.5cm 0.5cm 0.5cm 0.5cm},clip]{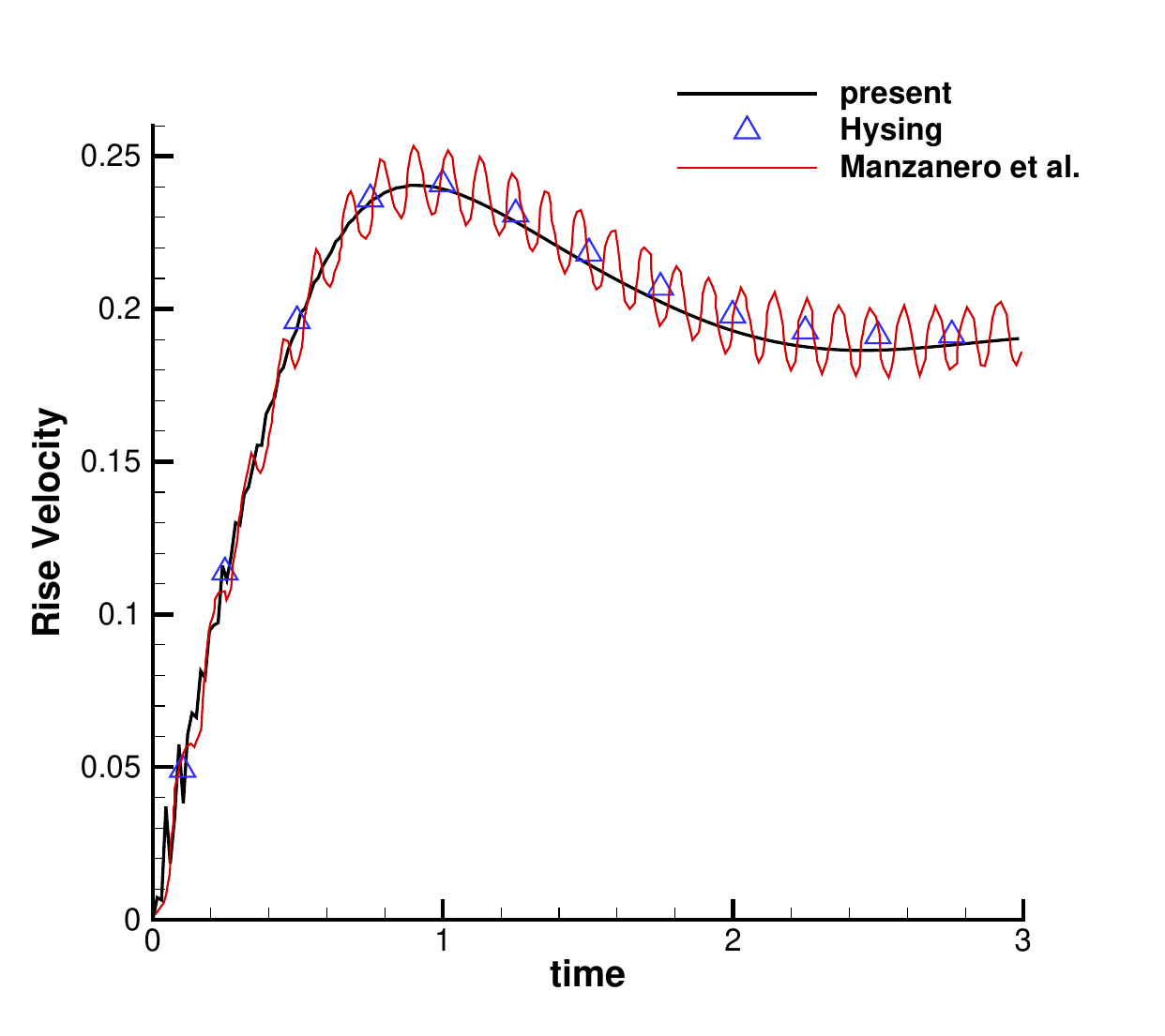}}
	\end{center}
	\protect\caption{Evolution of the centre of mass (a) and rising velocity (b) for the case 1 rising bubble problem. The present solution is plotted using the black solid lines. The reference solution from \cite{hysing2009quantitative} is represented by the blue triangle. The solution from \cite{manzanero2020entropy} is shown by the red lines.  
    \label{fig:risingBubbleFirstPlot}}	
\end{figure}

For Case 2, as shown in Fig.~\ref{fig:risingBubbleSecond}, the bubble undergoes significant deformation due to the combined effects of low surface tension and a large disparity in material properties across the interface. The interfacial topological changes calculated by the proposed method are presented in Fig.~\ref{fig:risingBubbleSecond}. At later times, the edges of the bubble tend to elongate into thin ligaments, a behaviour caused by the relatively weak surface tension forces compared to gravitational forces. For comparison, the reference solution obtained by the MULES method \cite{gamet2020validation} using an incompressible flow solver and a finer mesh size of $h=1.5625\times 10^{-3}$ is also included. It is observed that the bubble shape obtained by the current method generally agrees with the reference solution. The observed disparity is due to the much coarser mesh used in the current simulation. The evolution of the centre of mass and the increase in velocity are presented in Fig.~\ref{fig:risingBubbleSecondPlot}. It can be seen that the current solution agrees well with the reference solution calculated by MULES \cite{gamet2020validation}. Unlike the oscillatory solution obtained in \cite{manzanero2020entropy, yoo2021homogeneous}, the proposed method does not produce spurious oscillations. These results confirm that the proposed method can solve incompressible multiphase flow accurately and robustly.     

\begin{figure}
	\begin{center}
	\subfigure[t=0.7 s]{\includegraphics[width=.24\textwidth,trim={11.5cm 2.5cm 11.5cm 1.5cm},clip]{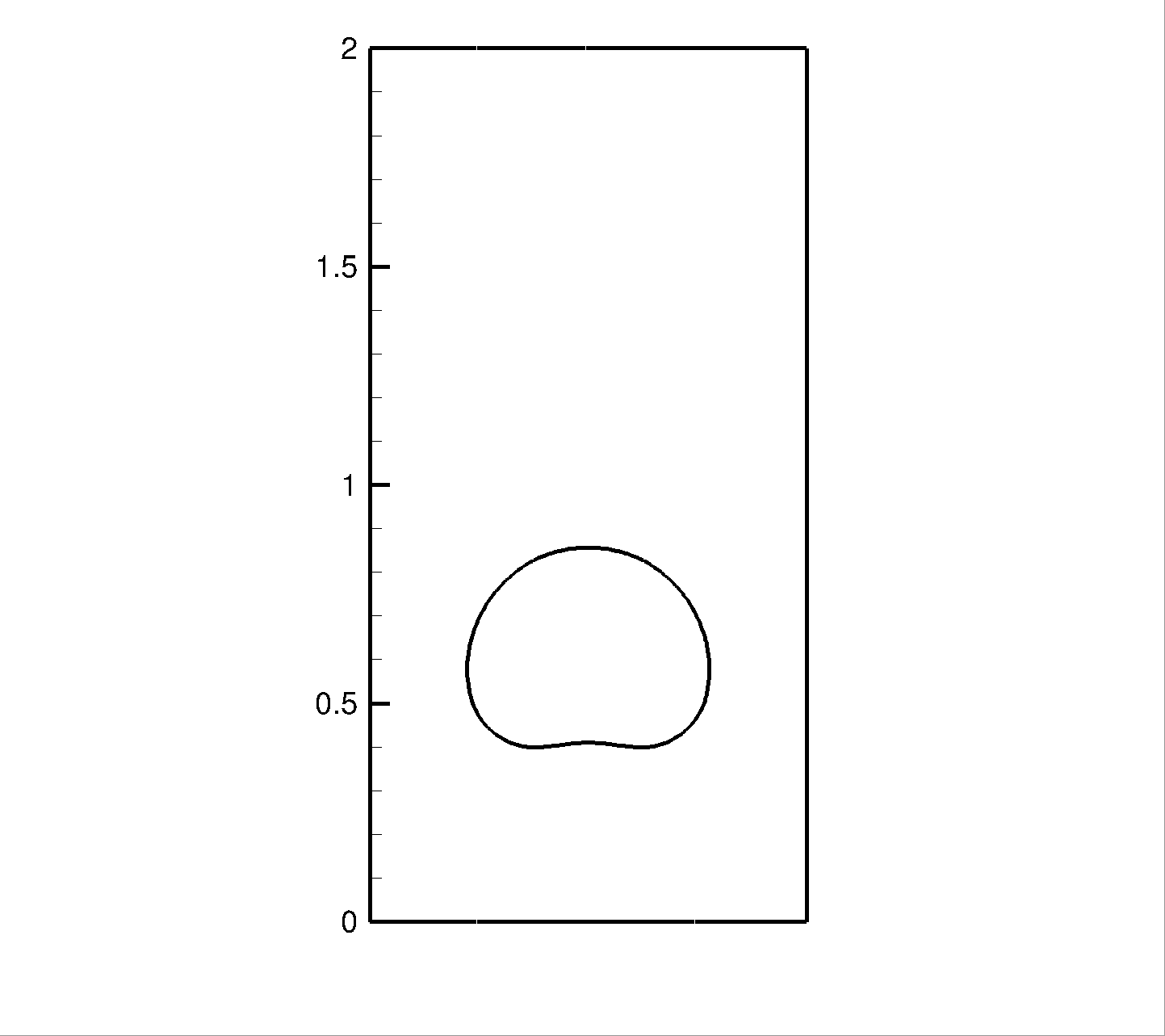}}
	\subfigure[t=1.2 s]{\includegraphics[width=.24\textwidth,trim={11.5cm 2.5cm 11.5cm 1.5cm},clip]{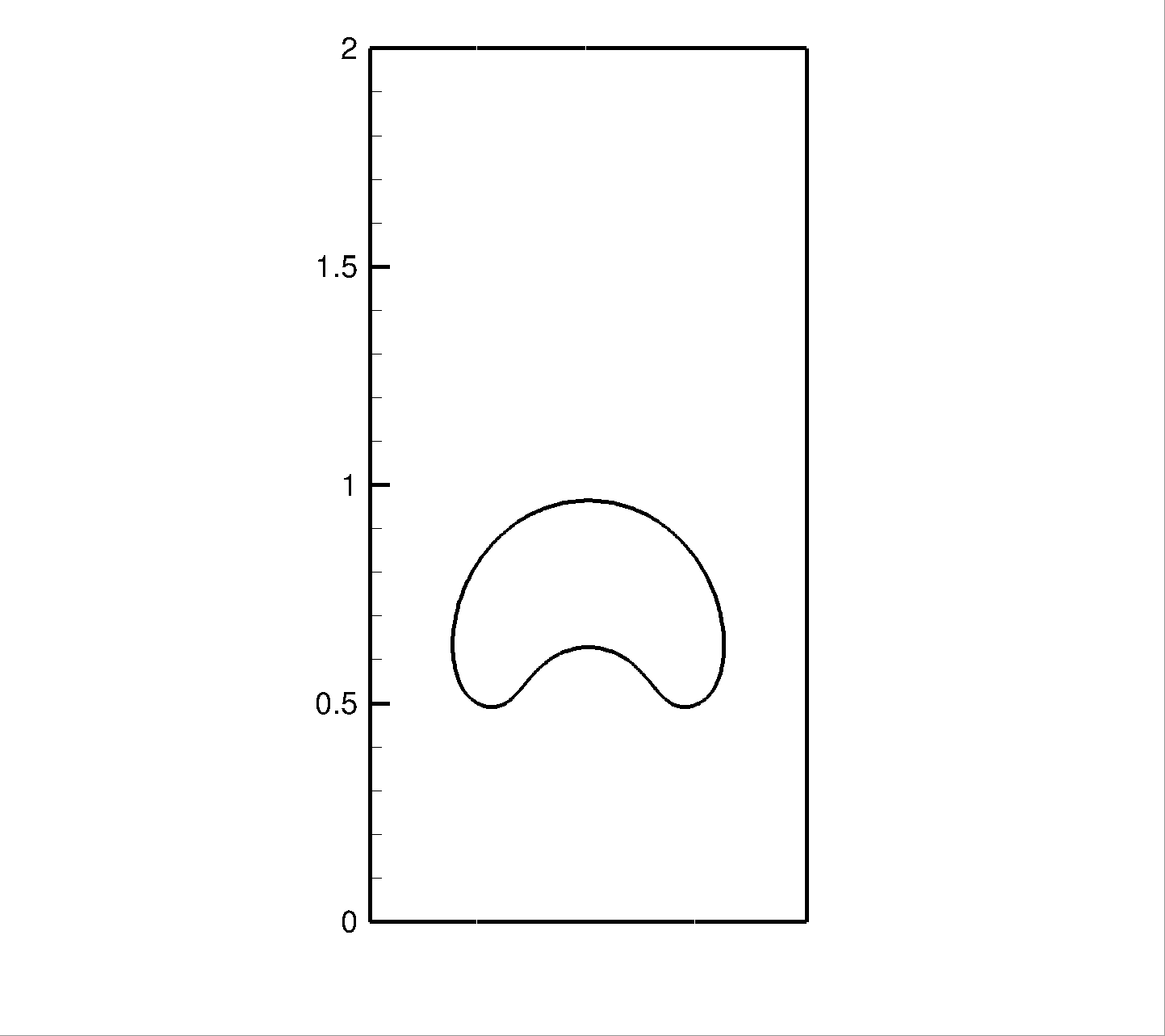}}
 	\subfigure[t=2.2 s]{\includegraphics[width=.24\textwidth,trim={11.5cm 2.5cm 11.5cm 1.5cm},clip]{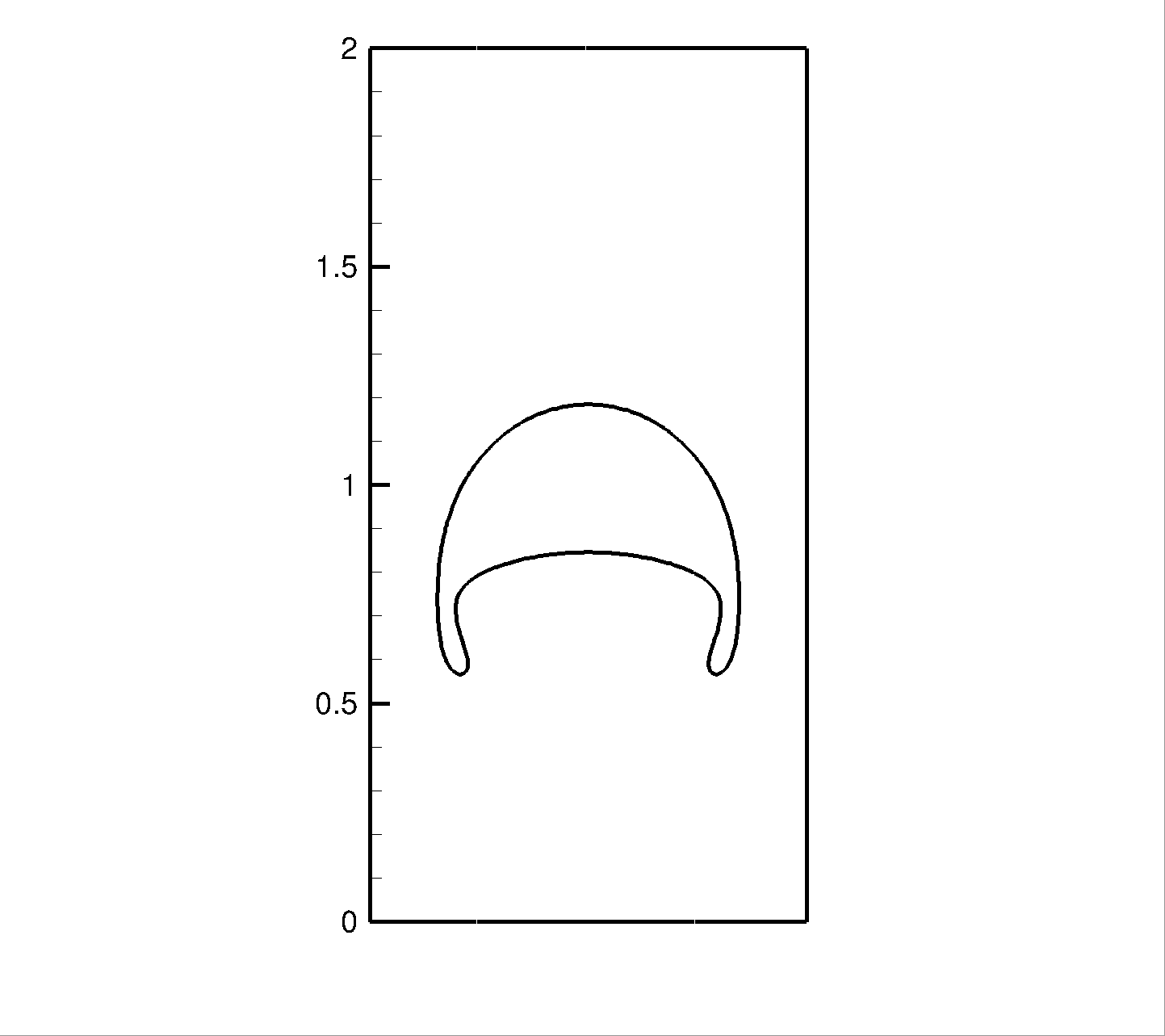}}
	\subfigure[t=3.0 s]{\includegraphics[width=.24\textwidth,trim={11.5cm 2.5cm 11.5cm 1.5cm},clip]{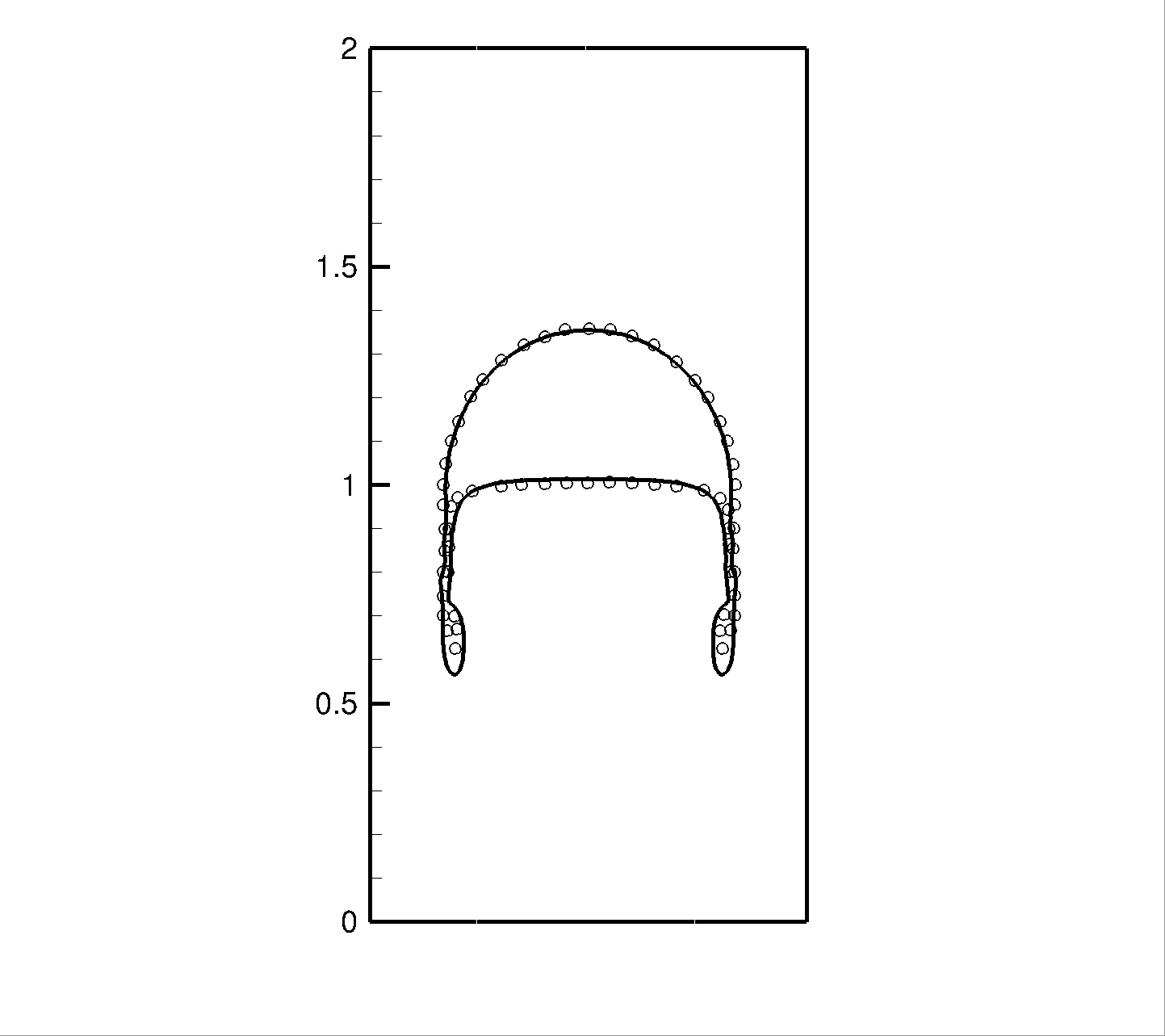}}
	\end{center}
	\protect\caption{Interface topological changes at different time instants for the case 2 rising bubble problem. The solid black line shows the interface calculated by the proposed method with mesh size $h=6.25\times 10^{-3}$. The circle symbols at the final time instant are the reference solution calculated by the MULES method with $h=1.5625\times 10^{-3}$ \cite{gamet2020validation}.
    \label{fig:risingBubbleSecond}}	
\end{figure}

\begin{figure}
	\begin{center}
	\subfigure[position of mass center]{\includegraphics[width=.42\textwidth,trim={0.5cm 0.5cm 0.5cm 0.5cm},clip]{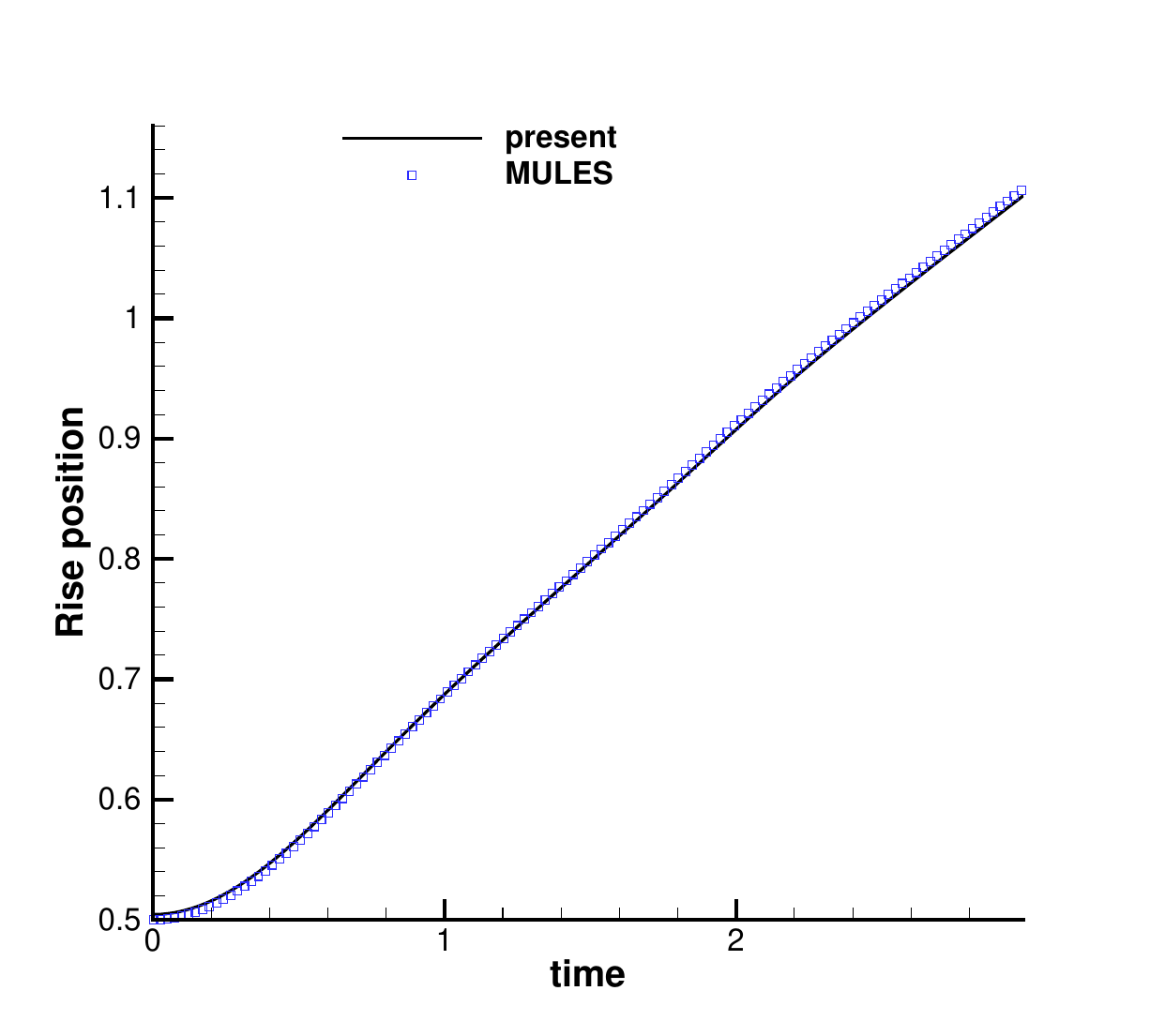}}
	\subfigure[rising velocity]{\includegraphics[width=.42\textwidth,trim={0.5cm 0.5cm 0.5cm 0.5cm},clip]{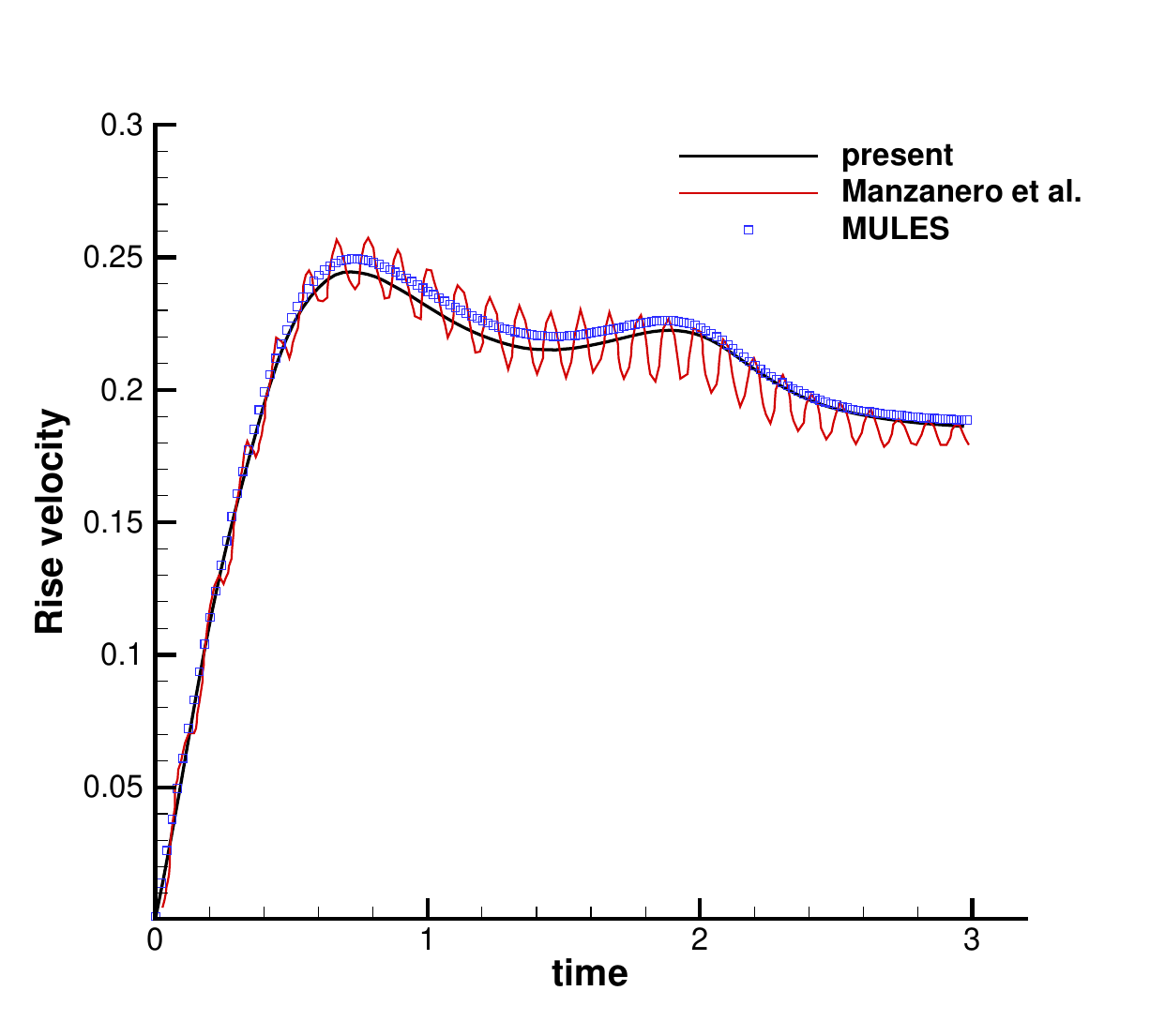}}
	\end{center}
	\protect\caption{Evolution of mass center (a) and rising velocity (b) for the case 2 rising bubble problem. The present solution is plotted in the black solid lines. The reference solution from \cite{gamet2020validation} is marked by the blue squares. The solution from \cite{manzanero2020entropy} is in red lines.  
    \label{fig:risingBubbleSecondPlot}}	
\end{figure}

\subsection{Air helium shock bubble interaction problem }
To validate the current solver's capability to handle high-speed compressible multiphase flows, we simulate the air-helium shock bubble interaction problem. The computational domain is $[-0.1,0.3]\times [-0.045,0.045]$ $\text{m}^2$. Hexahedral meshes of size $h=1.78\times 10^{-4}$ are used for the simulation. A stationary helium bubble with a diameter of 0.5 cm is initially positioned at the center of the domain. A moving shock wave with $Ma=1.22$ in the air is initially prescribed. Transmissive boundary conditions are applied on the left and right boundaries, while reflective boundary conditions are applied on the top and bottom boundaries. The thermodynamic properties of helium and air are provided in Table~\ref{tab:waterThermal}. In this case, viscosity, heat transfer, gravity, and surface tension are neglected.   

The numerical solution for the magnitude of the density gradient fields and the experimental results of \citet{haas1987interaction} are presented in Fig.~\ref{fig:shockHelium}. It is observed that the large-scale wave structures and interface evolution resolved by the numerical simulation agree well with the experimental results. This confirms that the proposed method can correctly solve material interfaces that interact with shock waves. Moreover, finer small-scale interface structures are resolved by the current solver, which is attributed to the high-resolution property of ROUND schemes. 

\begin{figure}
	\begin{center}
    \includegraphics[width=.8\textwidth,trim={0.0cm 0.5cm 0.5cm 0.5cm},clip]{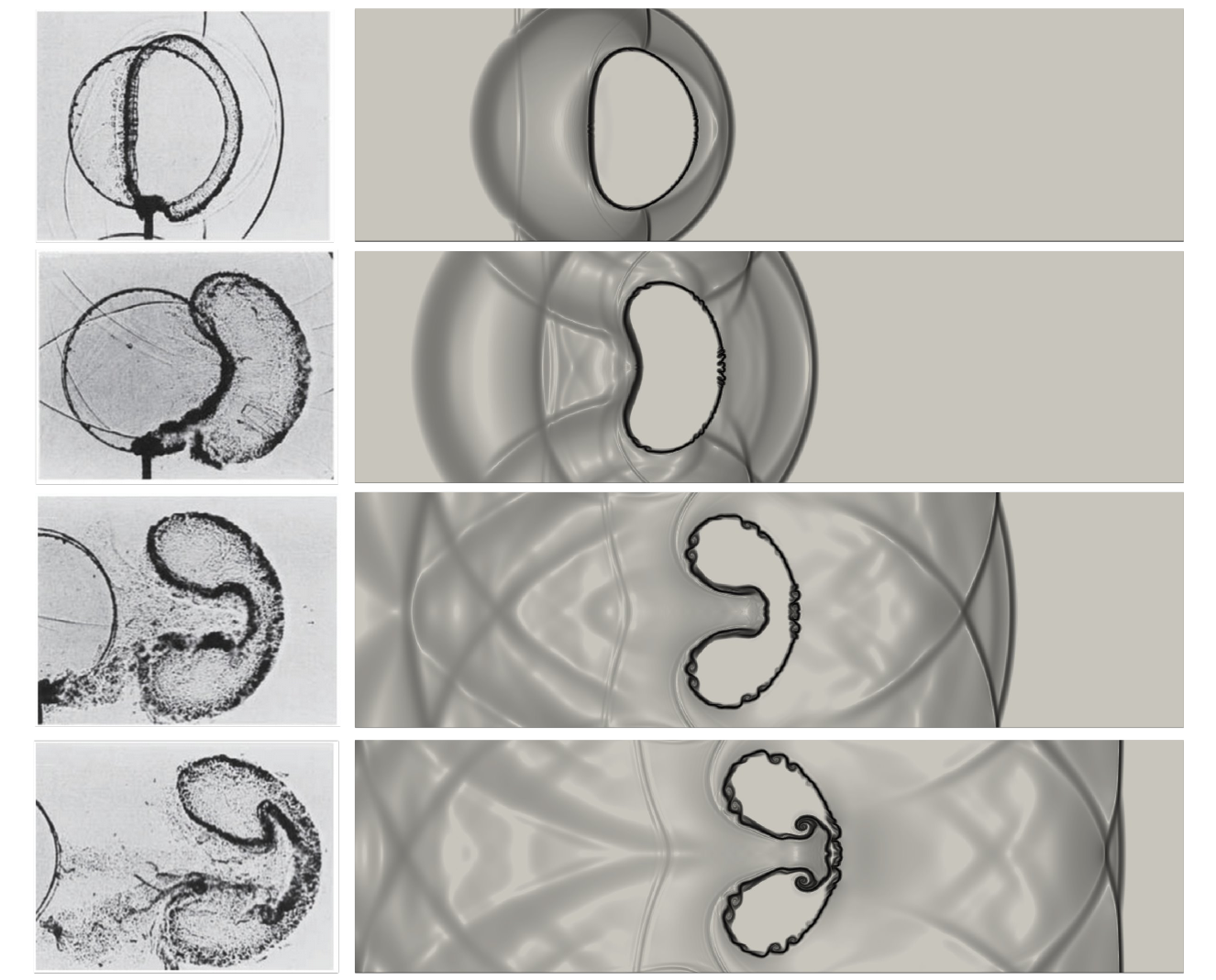}
	\end{center}
	\protect\caption{Numerical solution of the magnitude of the density gradient fields (right panel) and experimental results of \citet{haas1987interaction} (left panel) for air helium shock bubble interaction problem. The snapshots at $t=102 \mu \text{s}$, $t=245 \mu \text{s}$, $t=427 \mu \text{s}$ and $t=674 \mu \text{s}$ are presented from top to bottom. 
    \label{fig:shockHelium}}	
\end{figure}

\subsection{Shock wave interaction with a water column}
In this subsection, we simulate a challenging case involving a normal shock interacting with a cylindrical water column. This case is based on the experiment conducted by \citet{igra2001numerical}, which illustrates shock dynamics in the presence of an obstacle and primary atomization at a high
Mach number. This problem has also been simulated by a low-dissipation solver developed by \citet{kuhn2021all}. Initially, a water cylinder with a diameter of 4.8 mm is set in the center of the computational domain. A shock wave of $Ma= 1.47$ is prescribed in the air. The surface tension is included with $\sigma=0.0757 \text{kg}/\text{s}^2$. The viscosity of water is $\mu_l=1.0\times 10^{-3}$ Pa$\cdot$s and the viscosity of air is $\mu_g=1.8\times 10^{-5}$ Pa$\cdot$s. Thermal conductivity is not considered in this simulation. The thermodynamic parameters for water and air are given in Table \ref{tab:waterThermal}.    

A two-dimensional simulation is considered here. Hexahedral meshes with size $h=7.3684\times 10^{-5}$ are used, which corresponds to 65 cells across the diameter. The numerical solutions of the magnitude of the density gradient field at different snapshots are presented in Fig.~\ref{fig:shockWater} where the original holographic interferograms from the experiment are also included for comparison purposes. The numerical solution for the primary features of the shock wave, as shown in Fig.~\ref{fig:shockWater}(a) and Fig.~\ref{fig:shockWater}(b), agrees well with the experimental results. A pair of vortices behind the cylinder is also resolved by the current method, as shown in Fig.~\ref{fig:shockWater}(b). Compared to the numerical solution in \cite{kuhn2021all}, which uses a finer grid with 100 cells across the diameter, the present solver can still resolve wave propagations within the water column. The numerical solutions and experimental results at later times are presented in Fig.~\ref{fig:shockWater}(c). The numerical solution shows that filaments form behind the water column because of the strong shear flows. Since the simulation is two-dimensional, the atomization, an inherently three-dimensional phenomenon, cannot be accurately reproduced. For the shape of the interface in the numerical solution at the later time, the back side of the water column becomes flattened, which agrees with the experimental observation. These results demonstrate that the proposed method can accurately simulate compressible multiphase flows at a high Mach number.    

\begin{figure}
	\begin{center}
	\subfigure[$t=16~\mu \text{s}$]{\includegraphics[width=.5\textwidth,trim={0.5cm 0.5cm 0.5cm 0.5cm},clip]{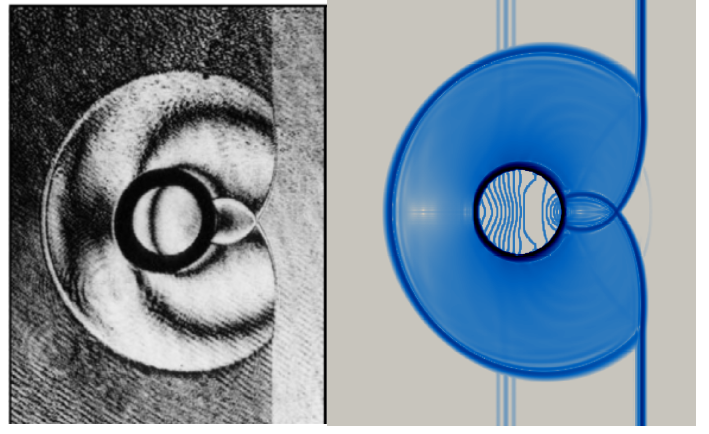}}
	\subfigure[$t=32~\mu \text{s}$]{\includegraphics[width=.5\textwidth,trim={0.5cm 0.5cm 0.5cm 0.5cm},clip]{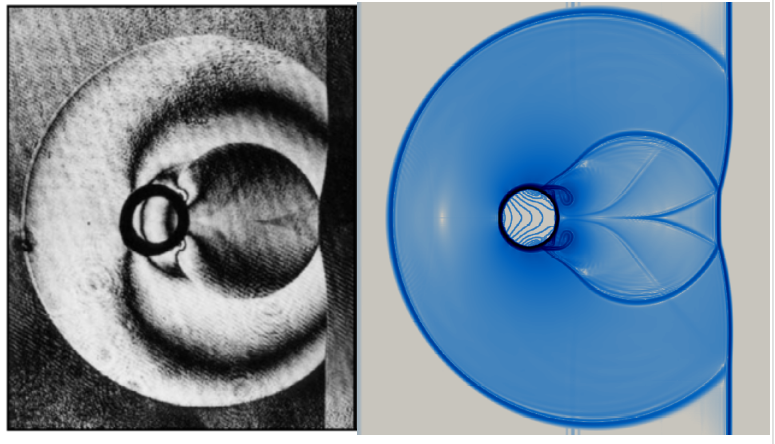}}
 	\subfigure[$t=206~\mu \text{s}$]{\includegraphics[width=.5\textwidth,trim={0.5cm 0.5cm 0.5cm 0.5cm},clip]{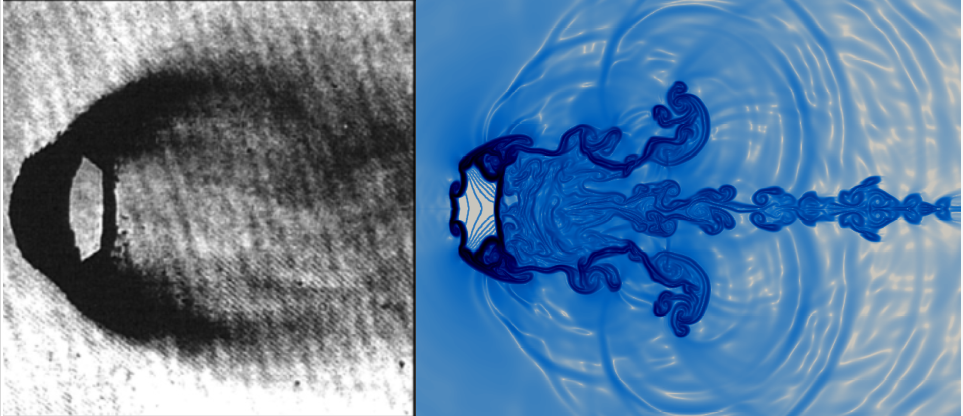}}
	\end{center}
	\protect\caption{Numerical solution of the magnitude of the density gradient fields (right panel) and holographic interferograms from the experiment by \citet{igra2001numerical} panel) for a shock wave interacting with a water column.
    \label{fig:shockWater}}	
\end{figure}

\subsection{Shock tube tests with multi-components and phase transitions}
In this subsection, several shock tube benchmark problems are used to verify the capability of the proposed hybrid approach in accurately resolving flow structures that contain shock waves, contact discontinuities, and rarefaction fans with phase transition effects. Conservative numerical schemes are generally required to resolve correct shock waves. The multicomponent mixtures consisting of liquid water, water vapour and air are considered. The thermodynamic parameters for these components are provided in Table \ref{tab:waterThermal}. The computational domain is $[0,1]$ with 100 mesh cells.

In this first shock tube test, the initial condition of a mixture far from the phase bounds is set as follows
\begin{equation*} \label{eq:initial}
(p, T, Y_1,Y_2,Y_3)=\left\{
\begin{array}{ll}
(0.2~\text{MPa},~T_{sat}(x_v p),~0.1,~0.2,~0.7),~~ x \leq 0.5, \\
(0.1~\text{MPa},~T_{sat}(x_v p),~0.1,~0.2,~0.7),~~\text{otherwise}.
\end{array}
\right.   
\end{equation*}
The initial condition indicates an initial pressure jump in the mixture, while the mixture is initially in the thermodynamic equilibrium state. The computation is run until $t=1 \text{ms}$. The numerical solutions calculated by the proposed solver are presented by the blue dashed lines, as shown in Fig.~\ref{fig:shockTube1}. The reference solutions marked by the black solid lines are produced by the density-based solver \cite{deng2020diffuse} using the same grid number. It can be seen that the proposed hybrid solver can correctly resolve the shock waves, contact discontinuity and rarefaction fan. The evaporation due to the shock compression and condensation due to expansion waves are also correctly resolved. It is also observed that the results produced by the proposed method are slightly more diffusive than the ones obtained by the density-based solver.

\begin{figure}
	\begin{center}
    \subfigure[mass fraction of liquid]   {\includegraphics[width=.42\textwidth,trim={0.5cm 10.5cm 0.5cm 10.5cm},clip]{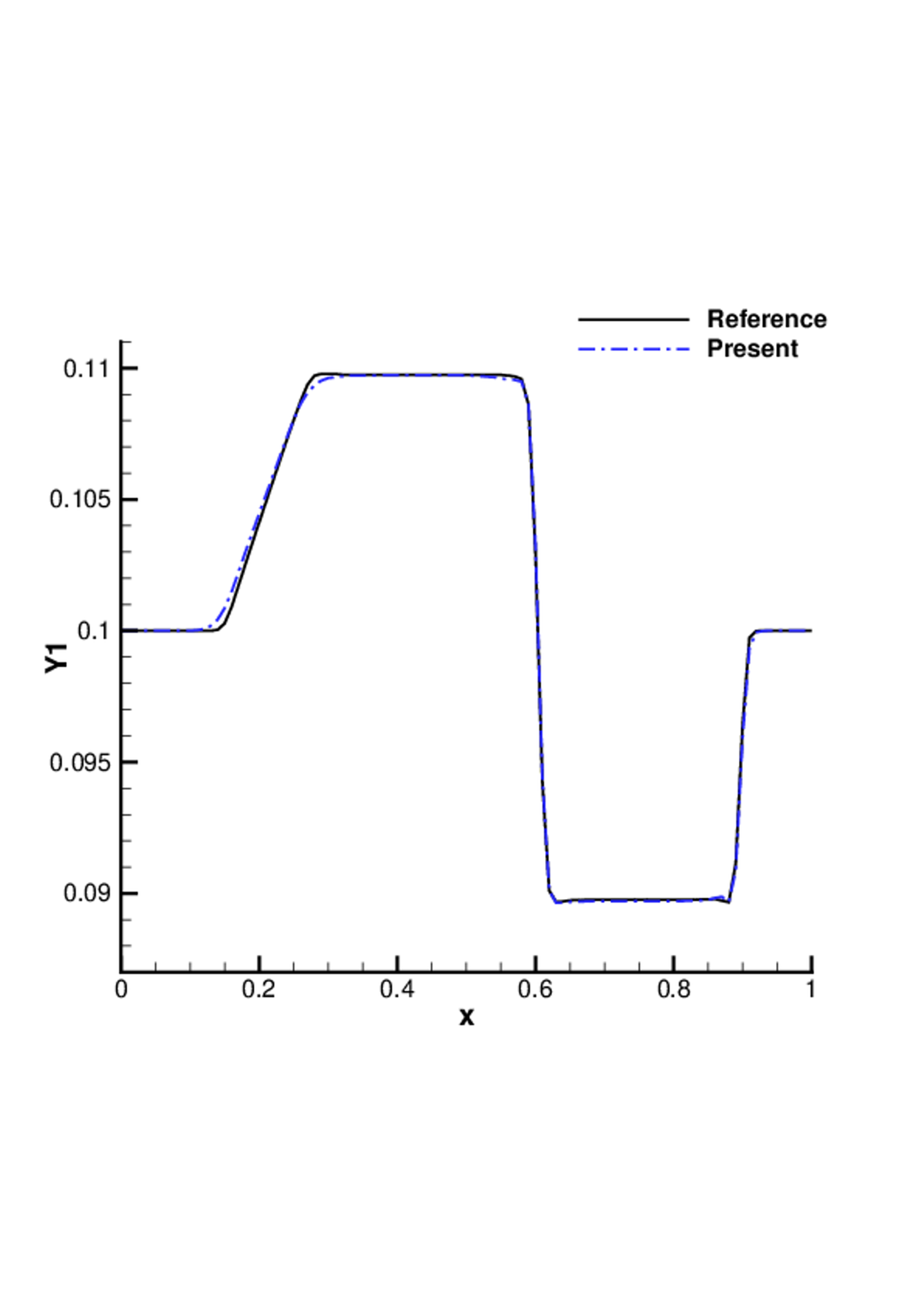}}
   	\subfigure[mass fraction of vapour]{\includegraphics[width=.42\textwidth,trim={0.5cm 10.5cm 0.5cm 10.5cm},clip]{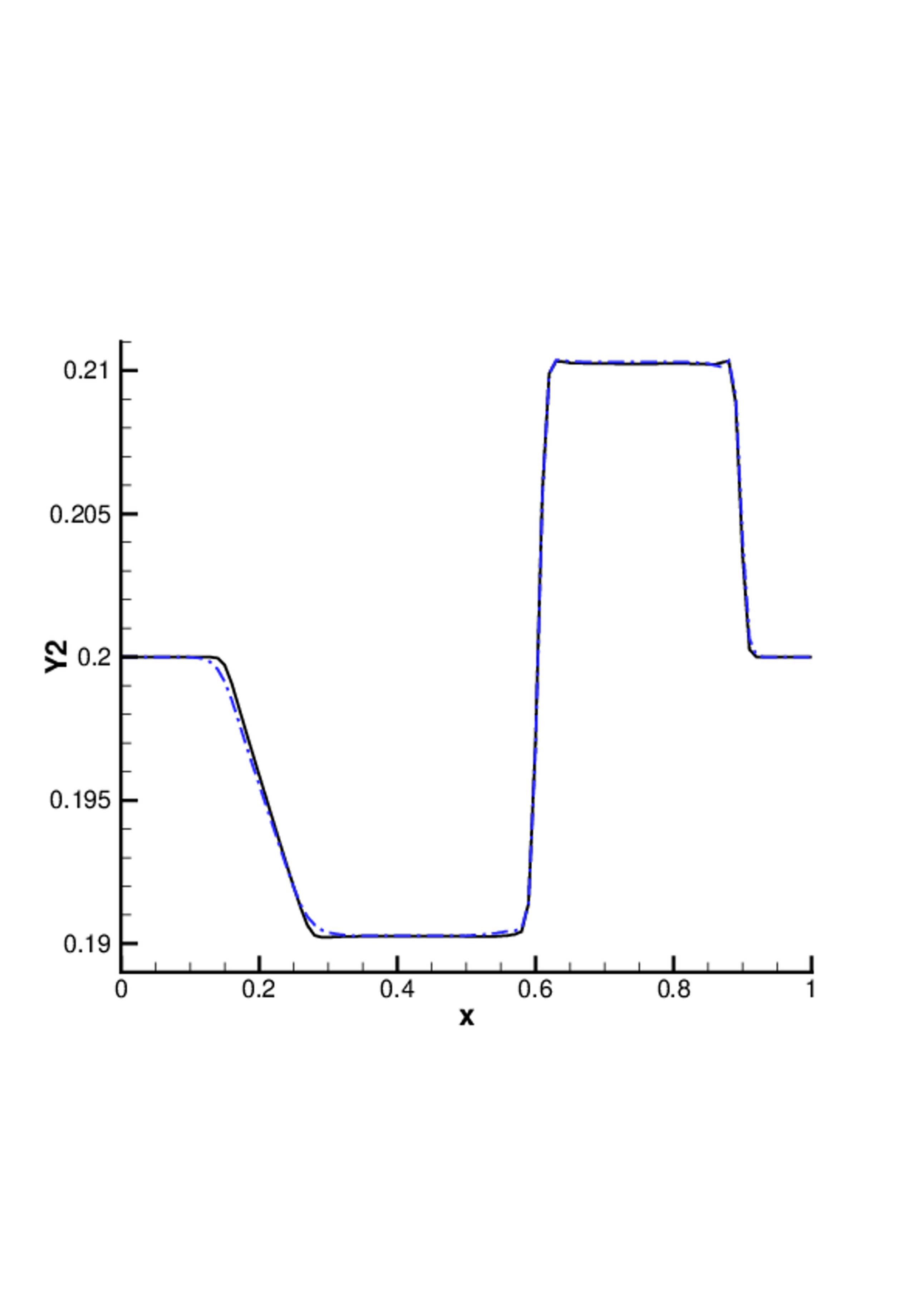}}
	\subfigure[pressure (Pa)]{\includegraphics[width=.42\textwidth,trim={0.5cm 10.5cm 0.5cm 10.5cm},clip]{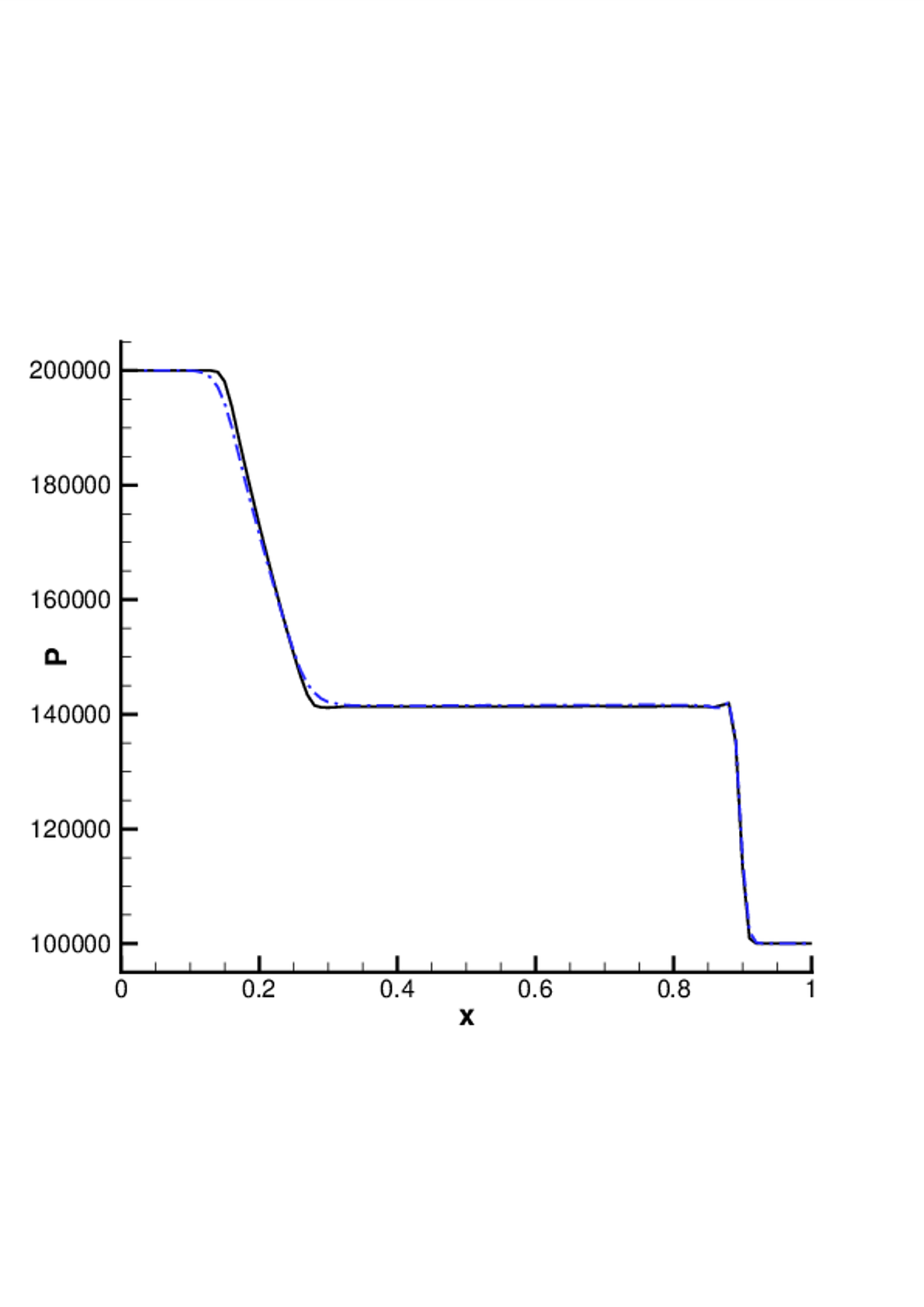}}
	\subfigure[temperature (K)]{\includegraphics[width=.42\textwidth,trim={0.5cm 10.5cm 0.5cm 10.5cm},clip]{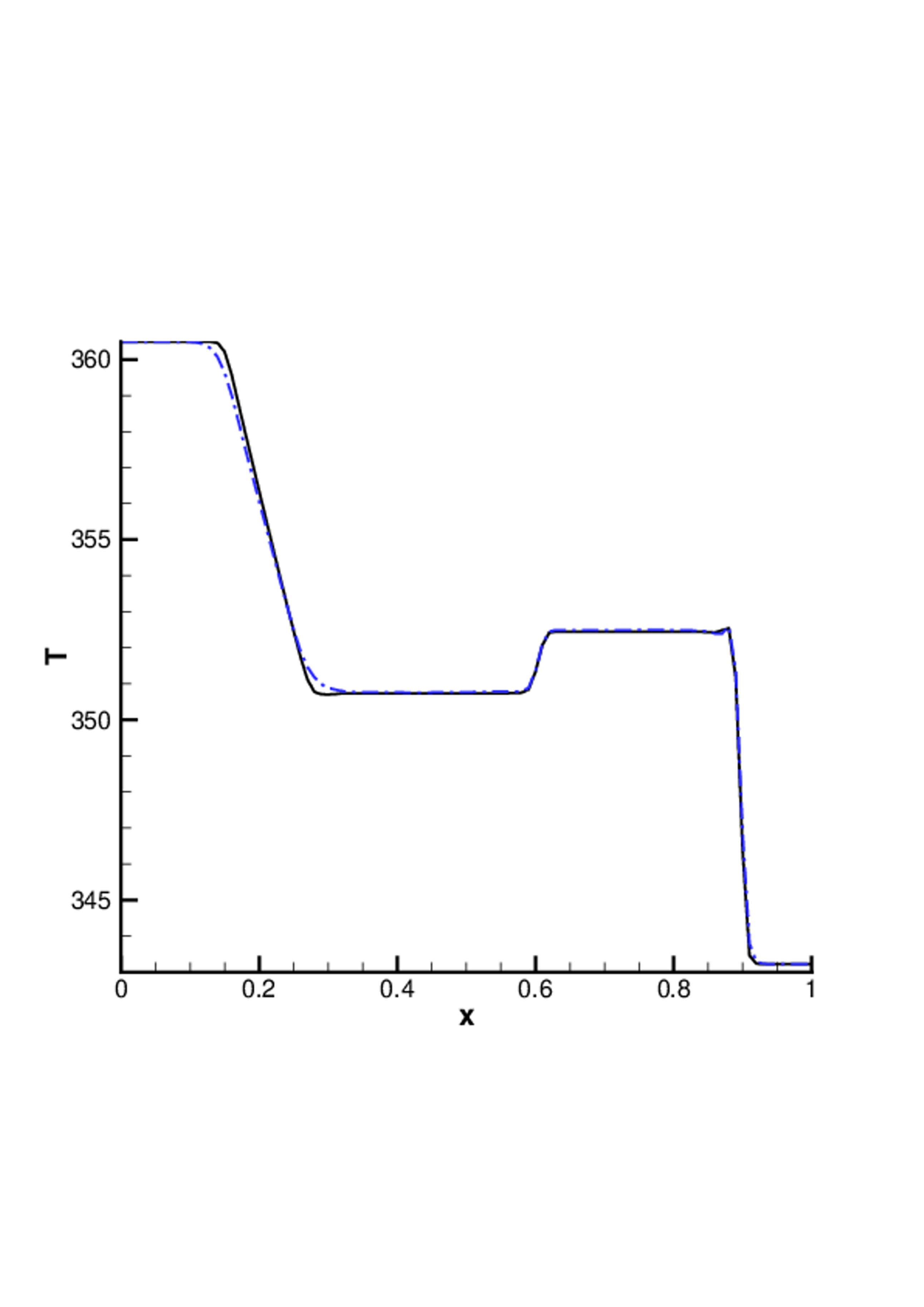}} 	
	\end{center}
	\protect\caption{Numerical solutions for the multi-component shock tube with the initial mixture far from the phase bounds. The reference solution by \citet{deng2020diffuse} is presented by the black solid lines. The solution by the proposed solver is plotted by the blue dashed lines. 
    \label{fig:shockTube1}}	
\end{figure}

The second shock-tube test considered here is initialised with an air-dominated mixture. The mass fraction of air is initially given as $Y_3=0.98$. An initial pressure ratio of 2 and an initial temperature of $T=293$K are provided in the computational domain. Then, the initial mass fraction of liquid water and water vapour is calculated by the thermodynamic equilibrium condition. The numerical solutions at $t=1$ms are presented in Fig.~\ref{fig:shocktube2}. The typical flow structures of the shock tube problem are again reproduced by the current numerical solver. Evaporation occurs alongside the shock wave, while condensation accompanies the rarefaction fan. As shown in the results, the solutions obtained by the proposed hybrid solver agree well with those calculated by the density-based solver.    

\begin{figure}
	\begin{center}
    \subfigure[mass fraction of liquid]    {\includegraphics[width=.42\textwidth,trim={0.5cm 0.5cm 0.5cm 0.5cm},clip]{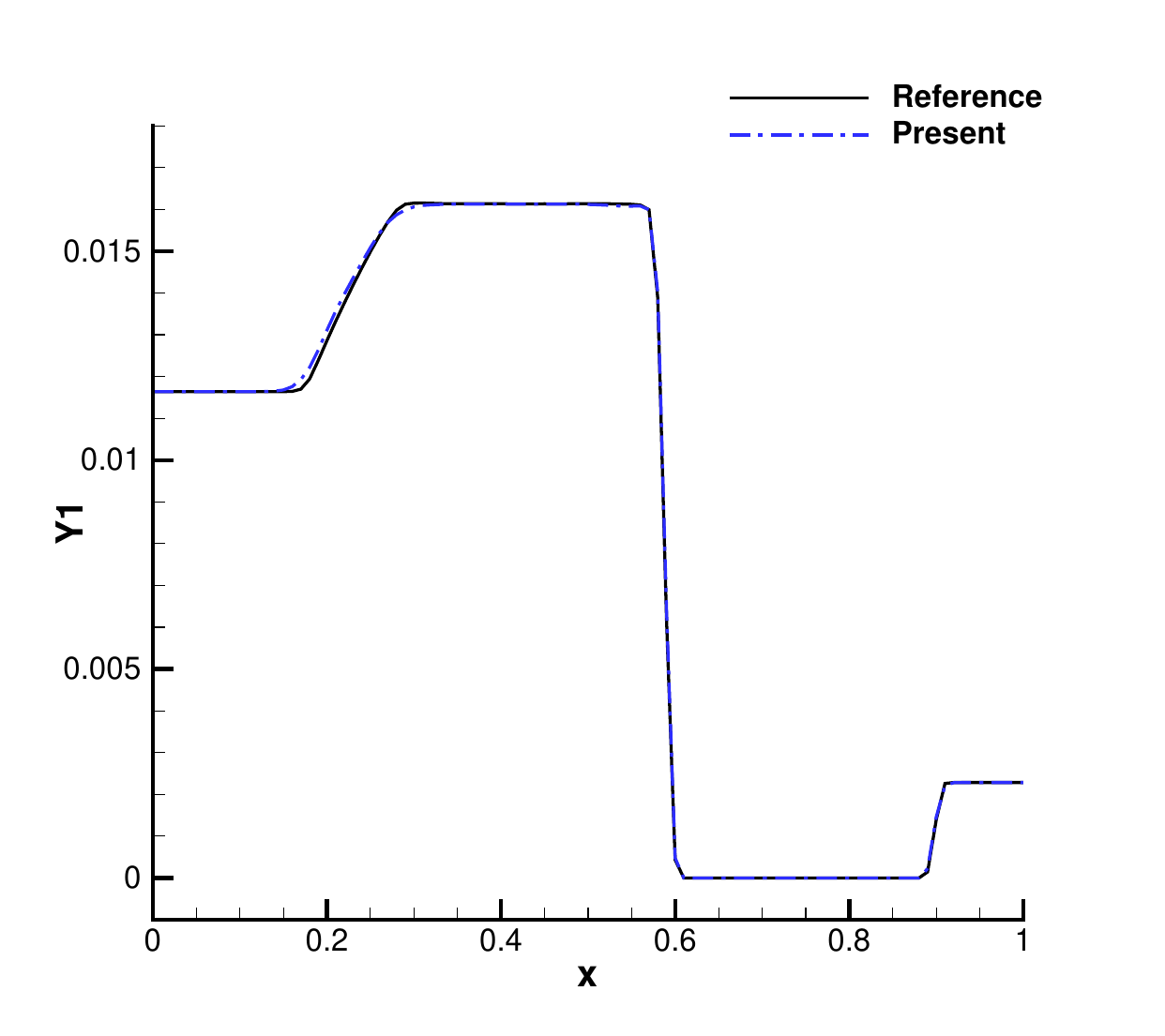}}
   	\subfigure[mass fraction of vapour]{\includegraphics[width=.42\textwidth,trim={0.5cm 0.5cm 0.5cm 0.5cm},clip]{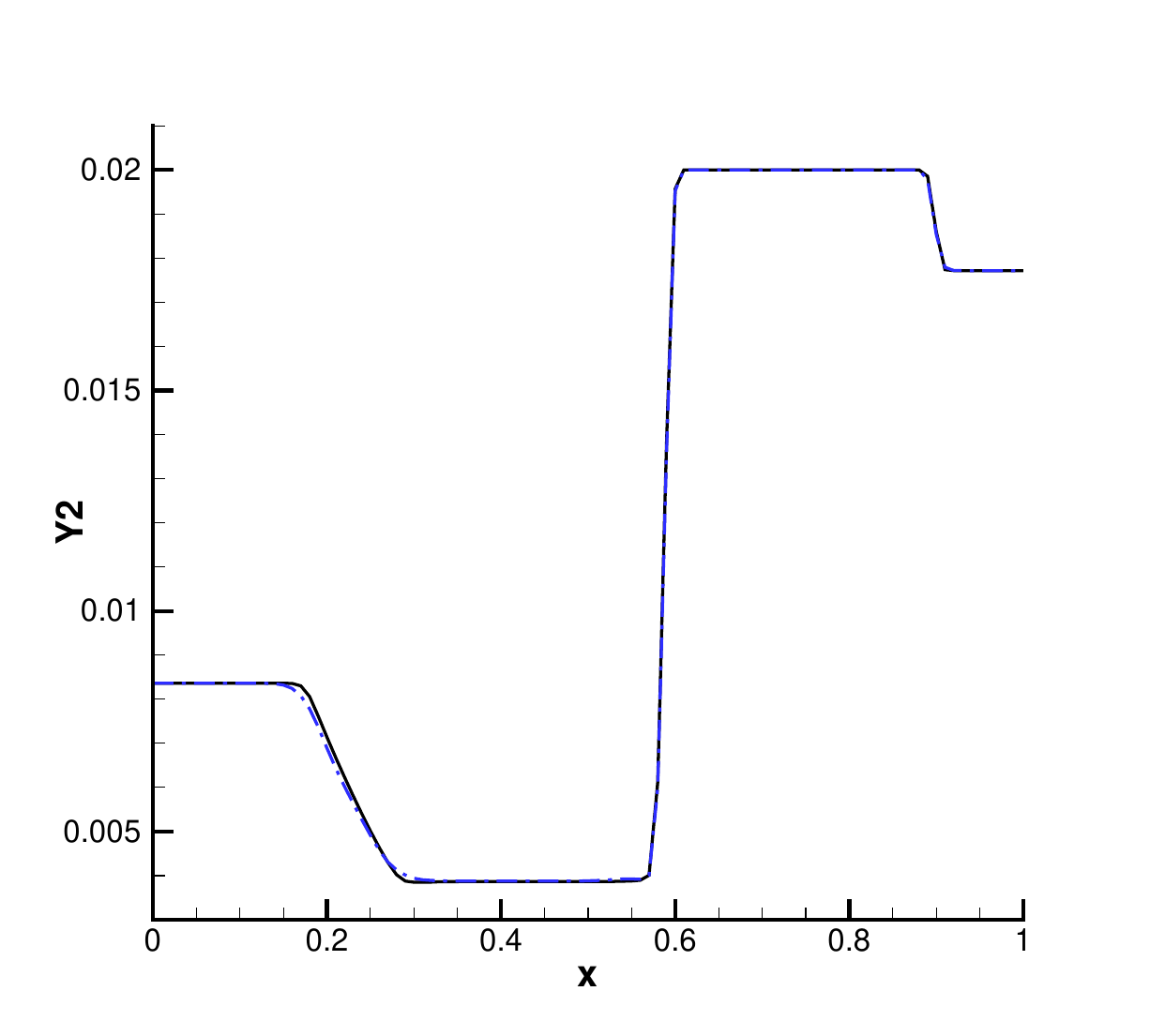}}
	\end{center}
	\protect\caption{Numerical solutions for the multi-component shock tube with air-dominated mixture. The reference solution by \citet{deng2020diffuse} is presented by the black solid lines. The solution by the proposed solver is plotted by the blue dashed lines.
    \label{fig:shocktube2}}	
\end{figure}

\subsection{Water liquid-vapour filled tube with a superheated region}
In this subsection, we consider a test case proposed previously \cite{pelanti2014mixture, de2019hyperbolic, pelanti2022arbitrary,demou2022pressure} designed to verify the capability of our numerical procedure to handle conditions relevant to boiling flows. A 1.7 m 1D tube is used as the computational domain. Following the initial conditions used by \citet{demou2022pressure}, the tube is filled with a water mixture, characterised by a water volume faction of $\alpha_1=0.9$ for $x<0.68$m and $\alpha_1=0.1$ for $x>0.68$m. A uniform pressure of $p=3.0104051\times 10^6$Pa is initially imposed in the tube. For the gas-dominated region $x>0.68$m, the temperature is set to the saturation condition. For the liquid-dominated region, the mixture is overheated by 2.0 K. Viscosity, thermal conductivity, and surface tension effect are neglected in the simulation. 

A uniform grid of 1000 mesh cells is used. The simulation is conducted until $t=6\times10^{-6}$s. Fig.~\ref{fig:superheated} shows numerical solutions obtained by the present method, along with the reference solution generated with a finer mesh of 4000 mesh cells. The numerical results indicate that the pressure on the liquid side rises to the saturation level, generating a pressure discontinuity. This pressure discontinuity results in two propagating waves that induce thermodynamic equilibrium between the two regions. Compared with the numerical solution by \citet{demou2022pressure}, the present method produces the same wave structures and a good agreement is achieved.

\begin{figure}
	\begin{center}
    \subfigure[pressure (Pa)]{\includegraphics[width=.42\textwidth,trim={0.1cm 0.5cm 0.5cm 0.5cm},clip]{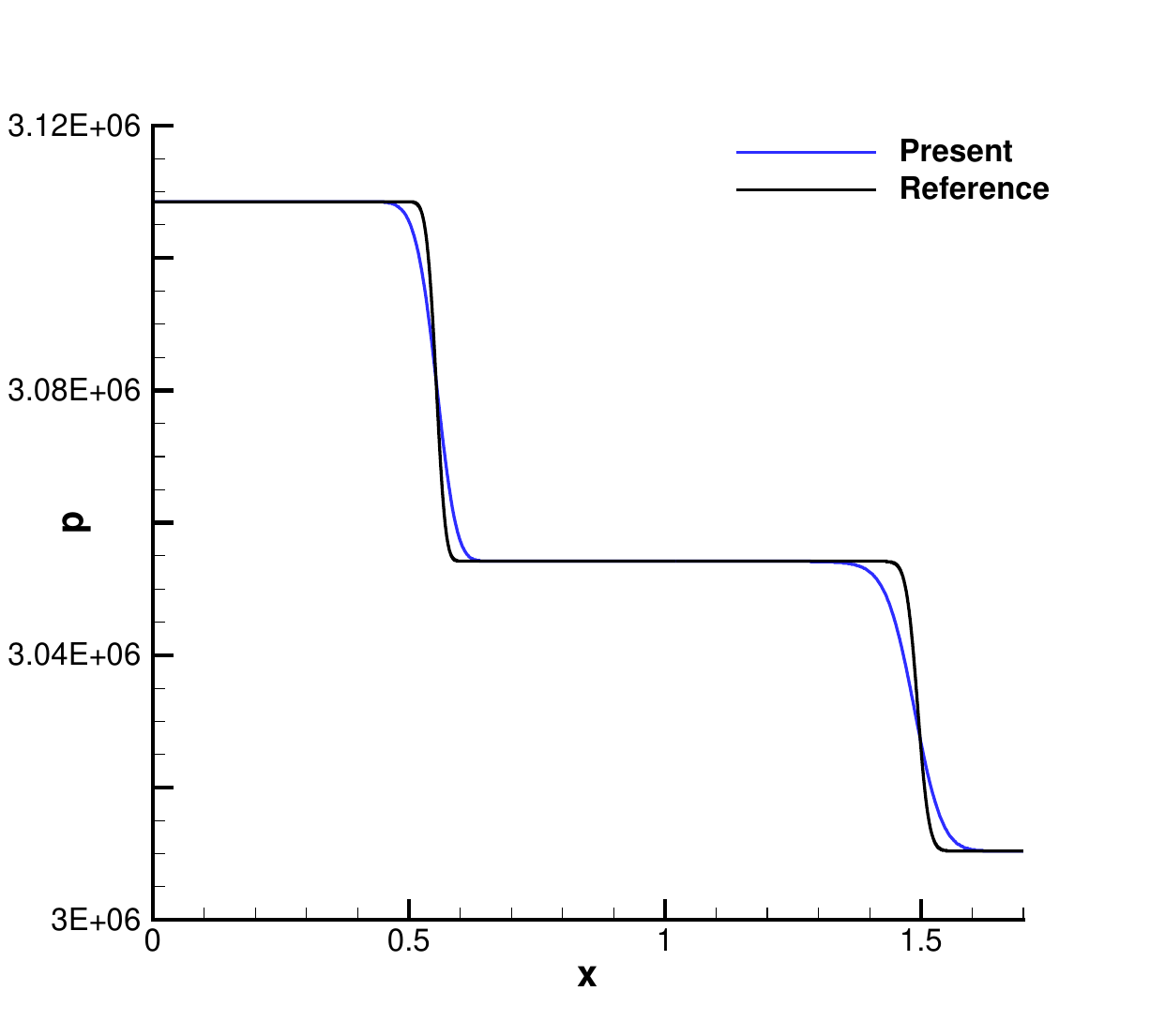}}
    \subfigure[volume fraction of vapour]    {\includegraphics[width=.42\textwidth,trim={0.1cm 0.5cm 0.5cm 0.5cm},clip]{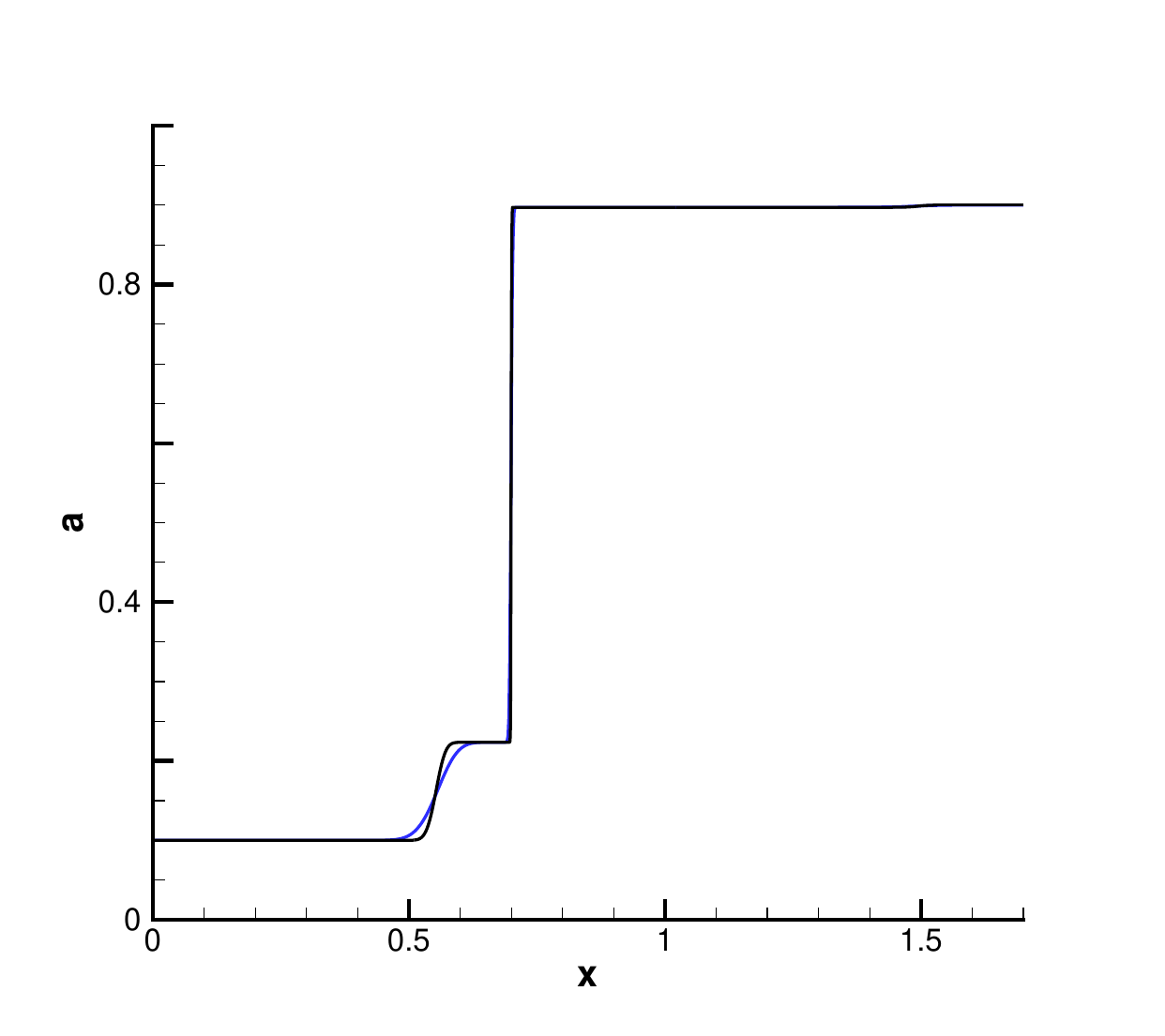}}
	\subfigure[temperature (K)]{\includegraphics[width=.42\textwidth,trim={0.1cm 0.5cm 0.5cm 0.5cm},clip]{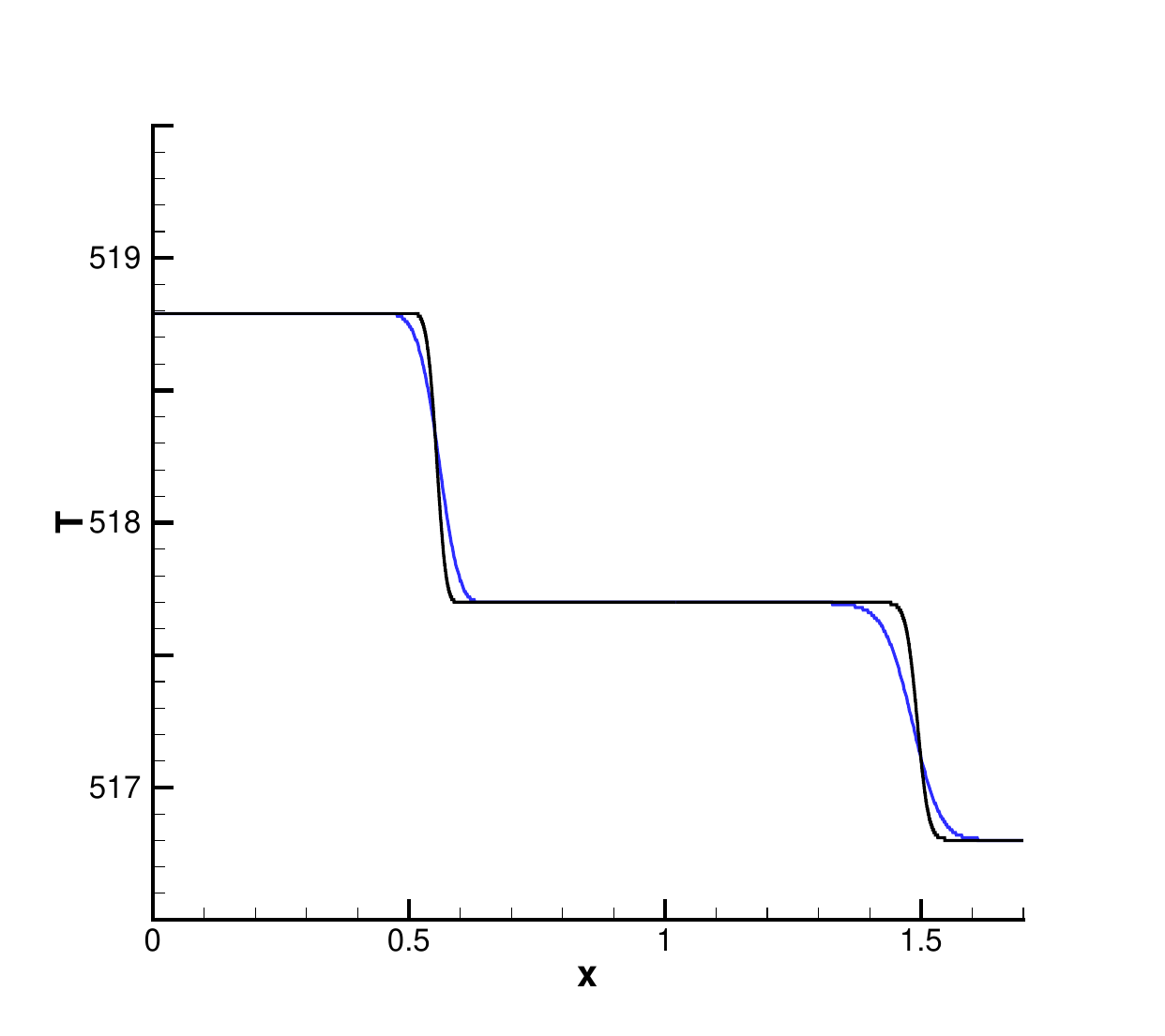}}
	\subfigure[mass fraction of vapour]{\includegraphics[width=.42\textwidth,trim={0.1cm 0.5cm 0.5cm 0.5cm},clip]{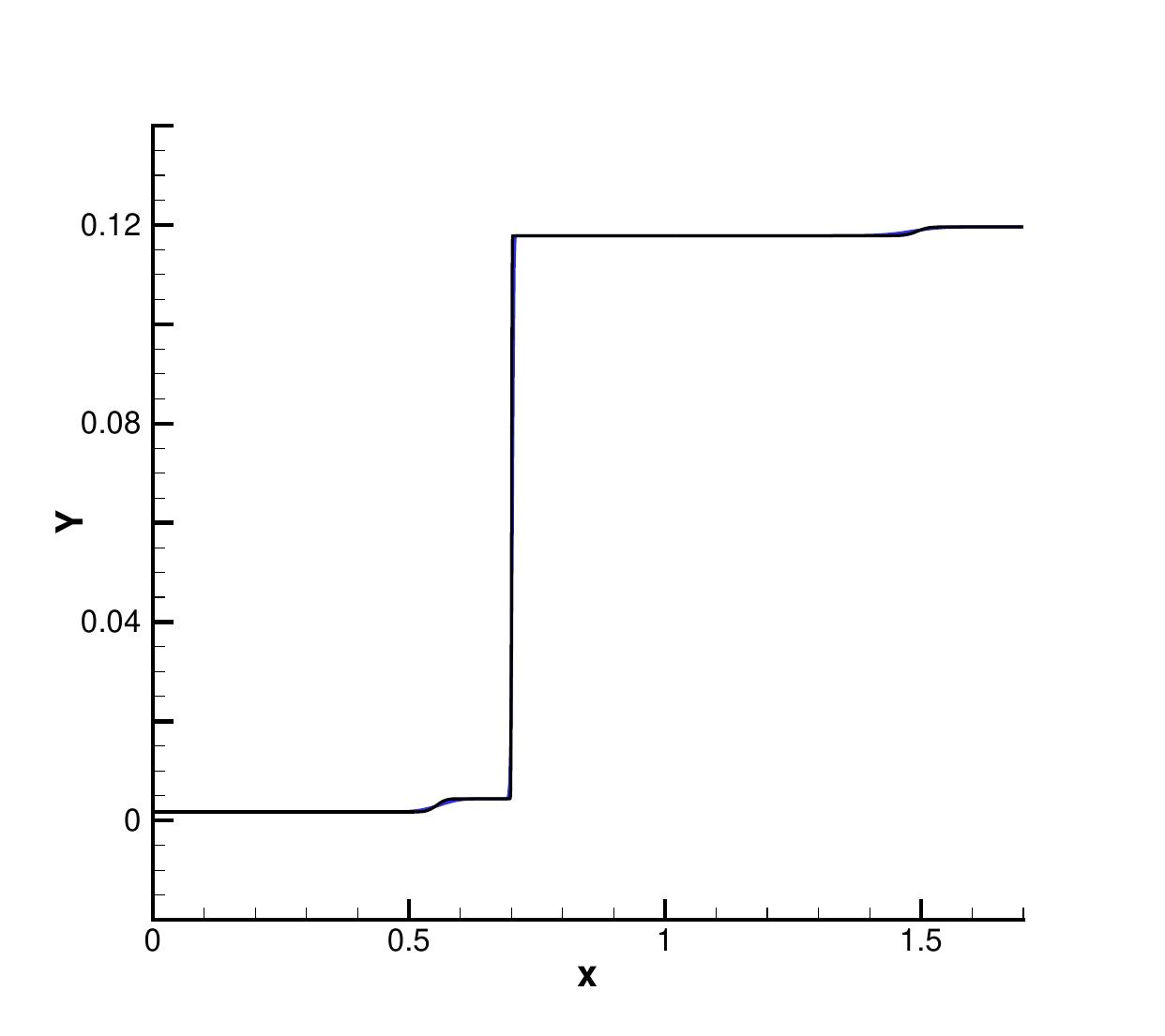}}
	\end{center}
	\protect\caption{Numerical solutions for the water liquid-vapour filled tube with a superheated region. The black solid lines present the reference solutions. 
    \label{fig:superheated}}	
\end{figure}

\subsection{2D nucleate boiling in water}
In order to demonstrate the capability of the proposed hybrid solver in solving the phase transition in the incompressible limit, in this section, a 2D boiling water test case is presented. The initial condition is set similar to that used by \citet{le2014towards}. The computational domain is 12 cm $\times$ 7 cm in which the lower half is filled with saturated liquid water ($\alpha=0.9999$) and the upper half is filled with saturated water vapour ($\alpha=10^{-4}$). The surface tension coefficient is $\sigma=0.073$N/m and the acceleration of gravity is $|\mathbf{g}|=9.81$m/$\text{s}^2$. The thermal conductivities of water in the liquid and vapour phases are $\lambda_1=0.6788\text{W}~\text{m}^{-1}~\text{K}^{-1}$ and $\lambda_1=0.0249\text{W}~\text{m}^{-1}~\text{K}^{-1}$, respectively. Similar to \citet{demou2022pressure}, the effects of viscosity and substrate wettability are neglected here. The initial pressure is $p=101325$ Pa. The initial temperature is then calculated as its saturation value. In this work, slip boundary conditions are imposed on all boundaries for the velocity field. The vertical and top walls are designated as adiabatic. The bottom wall is slowly heated, and a time-varying temperature is applied:
\begin{equation*} 
T(t)=\left\{
\begin{array}{ll}
T_{sat}+\Bigl(\frac{t}{0.15} \Bigl)\Delta T,~~ t \leq 0.15\text{s}, \\
T_{sat}+\Delta T,~~t > 0.15\text{s},
\end{array}
\right.   
\end{equation*}
where $T_{sat}$ is the saturation temperature. A uniform hexahedral grid of $210 \times 500$ is used. The numerical solutions of the volume fraction of water vapour are presented in Fig.~ \ref{fig:boiling}. It can be seen that in the early stages, a water vapour film starts to form at the bottom wall. A large elongated vapour bubble is created in the center as a result of the wall effects. The rising bubbles are created as a result of a combination of convective and surface tension effects. The phenomenological description of the current simulation results is similar to those reported in \citet{le2014towards,demou2022pressure}. This test case demonstrates that the proposed solver is capable of simulating phase transition and multi-physical processes at low Mach number.  

\begin{figure}
	\begin{center}
    \subfigure[t=0 s]{\includegraphics[width=.2\textwidth,trim={6.1cm 3.5cm 20.5cm 3.5cm},clip]{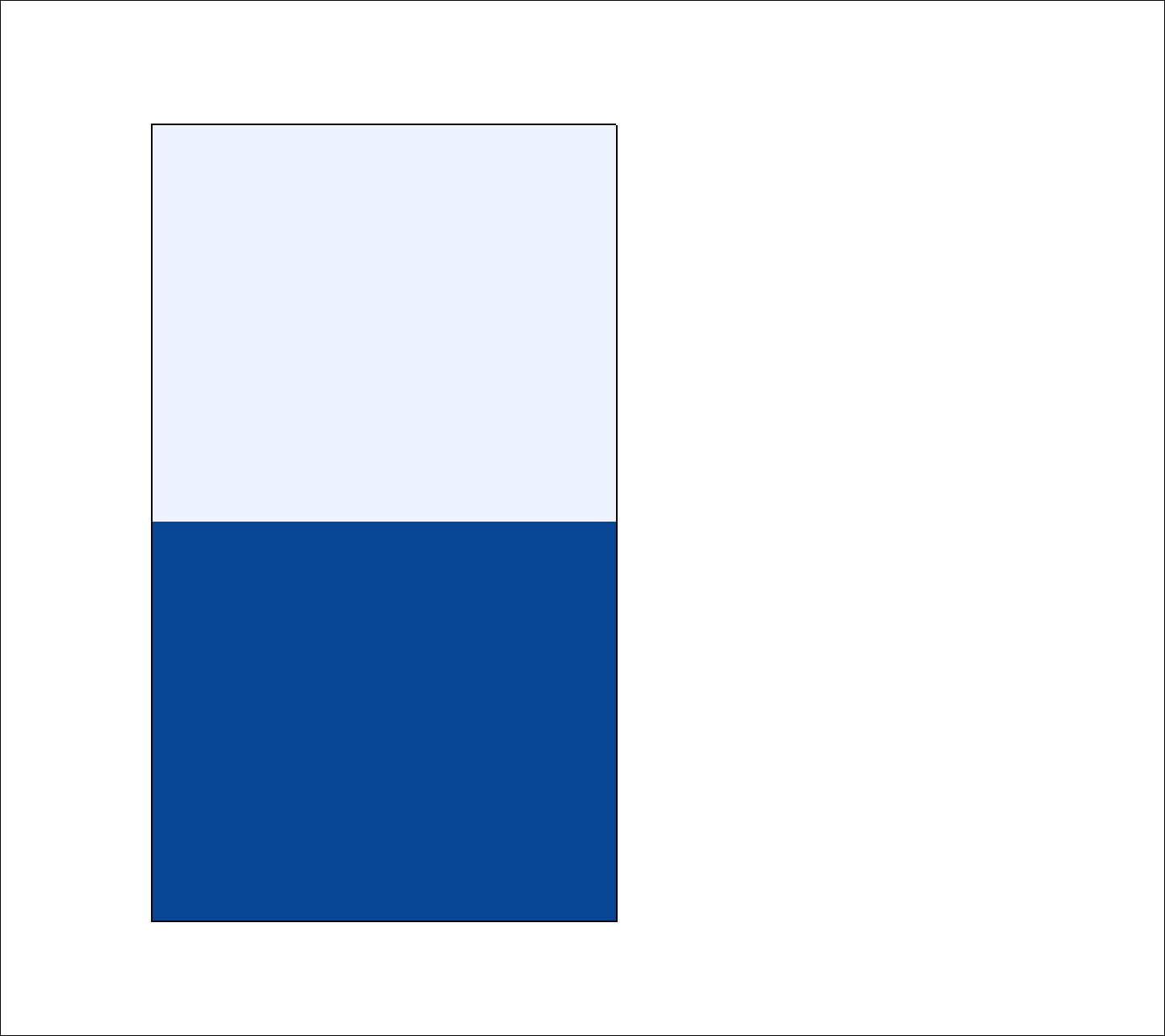}}
    \subfigure[t=250 ms]    {\includegraphics[width=.2\textwidth,trim={6.1cm 3.5cm 20.5cm 3.5cm},clip]{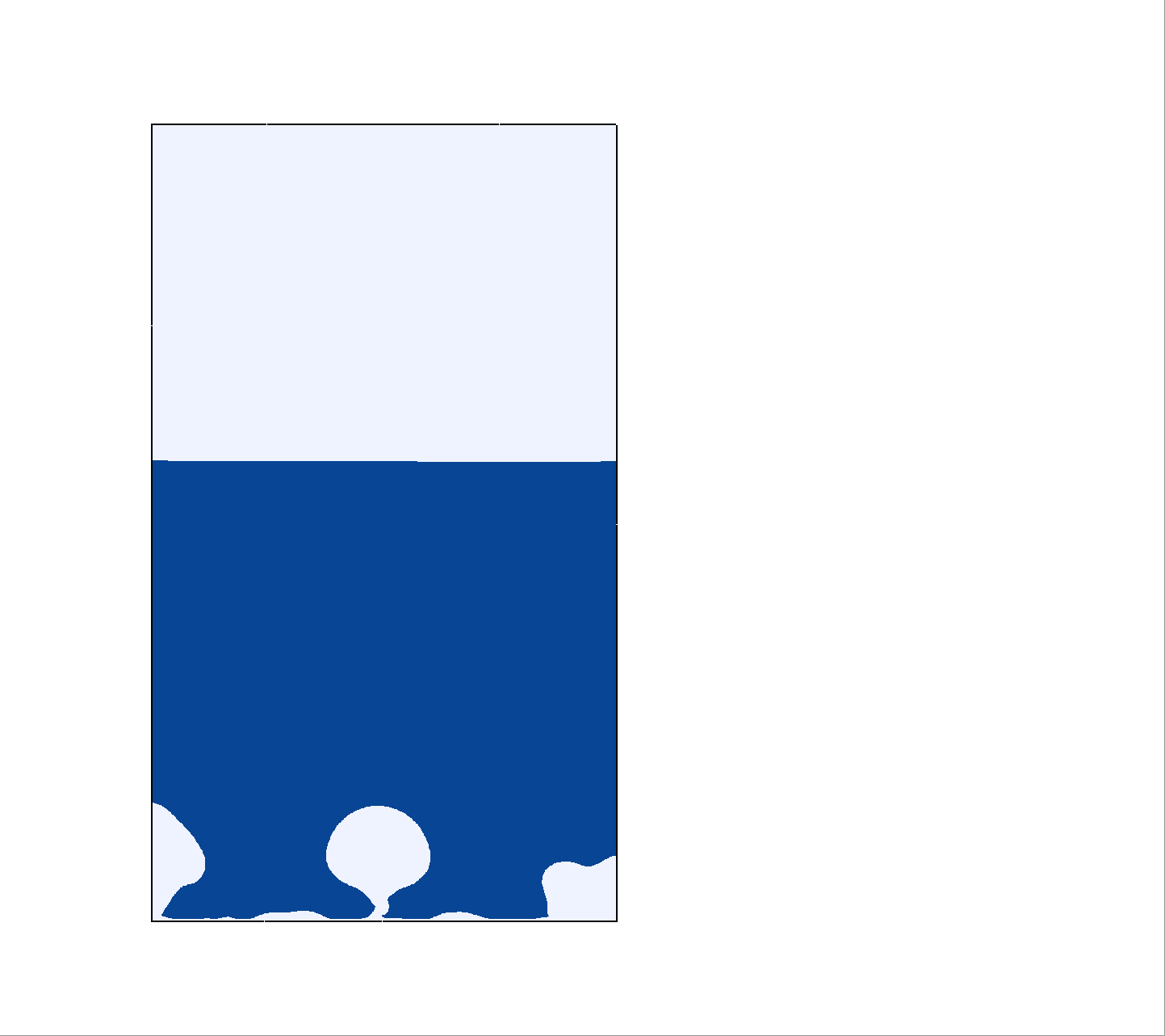}}
	\subfigure[t=300 ms]{\includegraphics[width=.2\textwidth,trim={6.1cm 3.5cm 20.5cm 3.5cm},clip]{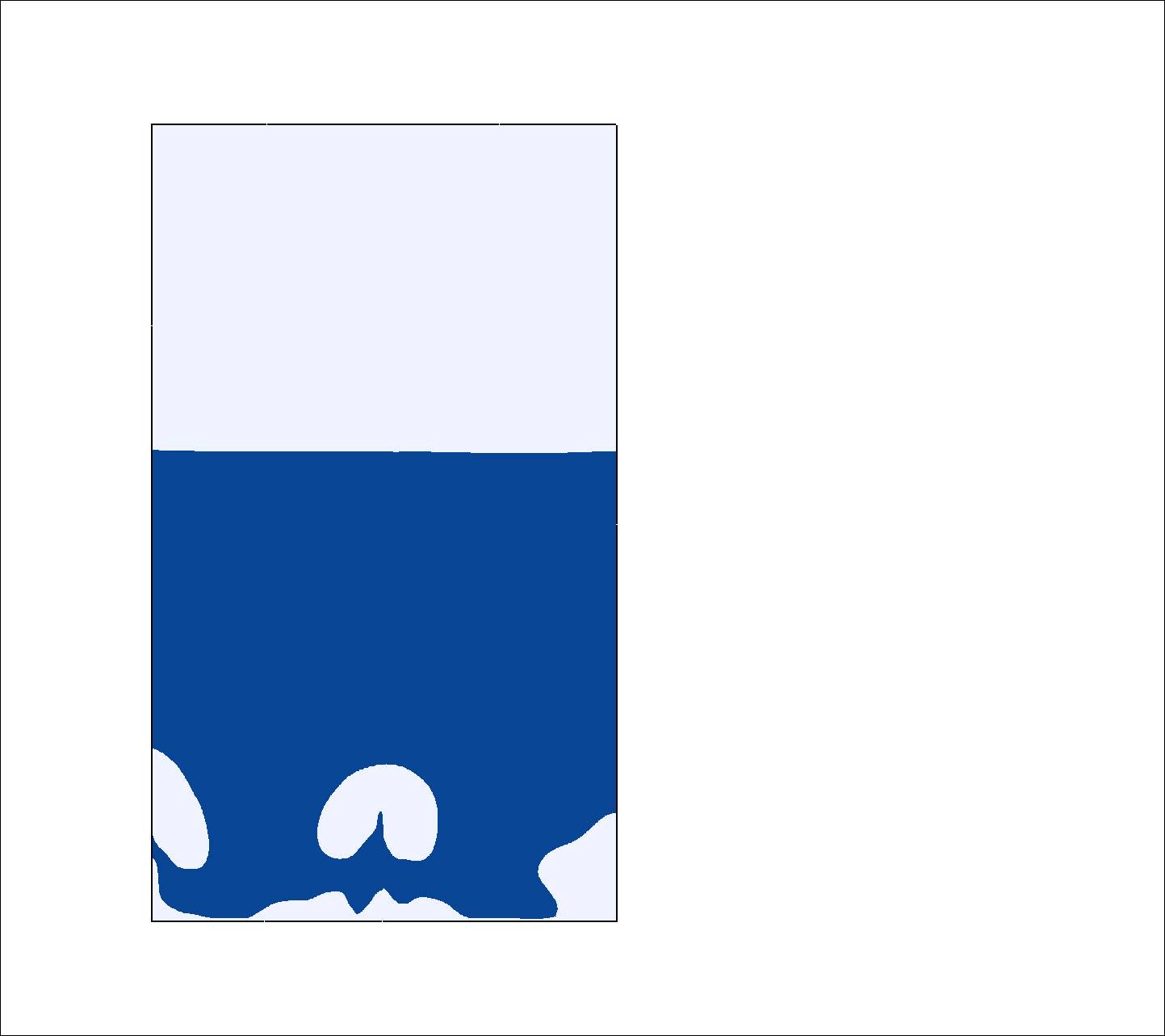}}
	\subfigure[t=350 ms]{\includegraphics[width=.2\textwidth,trim={6.1cm 3.5cm 20.5cm 3.5cm},clip]{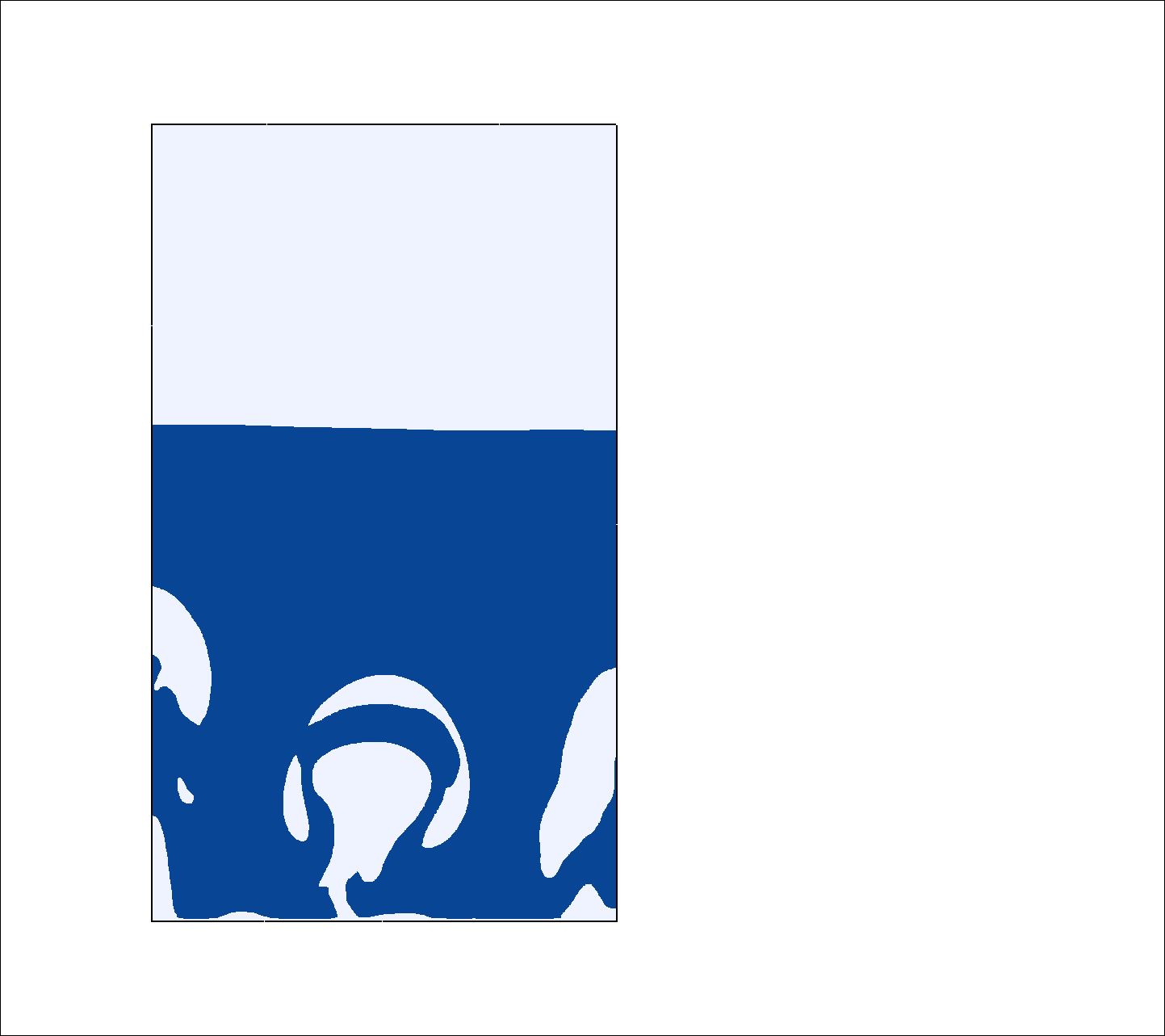}}
	\end{center}
	\protect\caption{Volume fraction of water vapour for the 2D nucleate boiling problem at different instances. 
    \label{fig:boiling}}	
\end{figure}

\subsection{Richtmyer-Meshkov instability with cavitation}
A shock-induced Richtmyer-Meshkov instability and cavitation bubbles are considered in this case. The computational domain is $[0,1]\times[0,1]$ $\text{m}^2$ with $500 \times 500$ hexahedral mesh cells. The slip boundary conditions are given at the boundaries. The liquid water-dominated mixture is filled in half of the domain of $y\leq0.36+0.15 \text{cos}(4\pi x)$. The aerogel-dominated mixture is filled in another half of the domain. An upward-moving planar shock of $Ma=1.7$ is initially positioned in the water-dominated mixture. The details of the initial condition for each physical variable are
\begin{equation*} \label{eq:initial}
(p, T,\bold{u}, Y_1,Y_2,Y_3)=\left\{
\begin{array}{ll}
(5\times10^8~\text{Pa},~420.7~\text{K},~(0,275)~\text{m}/\text{s},~0.99998,~1\times10^{-5},~1\times10^{-5}),~~ y \leq 0.2, \\
(1\times10^5~\text{Pa},~298~\text{K},~(0,0)~\text{m}/\text{s},~0.99998,~1\times10^{-5},~1\times10^{-5}),~~0.2<y\leq 0.36+0.15 \text{cos}(4\pi x), \\
(1\times10^5~\text{Pa},~298~\text{K},~(0,275)~\text{m}/\text{s},~1\times10^{-5},~1\times10^{-5},~0.99998),~~y> 0.36+0.15 \text{cos}(4\pi x)
\end{array}
\right..   
\end{equation*}
The initial distribution of the volume fraction field of liquid water is presented in Fig.~\ref{fig:RMinstability}(a). The thermodynamic parameters for modelling liquid water, water vapour and aerogel are provided in Table~\ref{tab:waterThermal}, resulting in an initial density ratio of approximately 10 between water and aerogel. Viscosity, thermal conductivity, surface tension, and gravity are neglected in the simulation. The numerical solutions regarding the volume fraction of liquid water at different instances are presented in Fig.~\ref{fig:RMinstability}. As shown in Fig.~\ref{fig:RMinstability}(b), two water jets form following the interaction of the upward-moving shock wave with the interface. As the shock wave continues upward, two jets grow and develop kelvin-Helmholtz-like shapes, as shown in Fig.~\ref{fig:RMinstability}(c)(d). Cavitation bubbles are also generated as rarefaction waves reflect from the upper interface. These simulation results are in agreement with the observations from the experimental results of \citet{strucka2023synchrotron}. This test case demonstrates the capability of the proposed method for simulation of the Richtmyer-Meshkov instability which involves shock waves and phase transition.       

\begin{figure}
	\begin{center}
    \subfigure[t=0 $\mu$s]{\includegraphics[width=.42\textwidth,trim={2.5cm 2.5cm 2.5cm 2.5cm},clip]{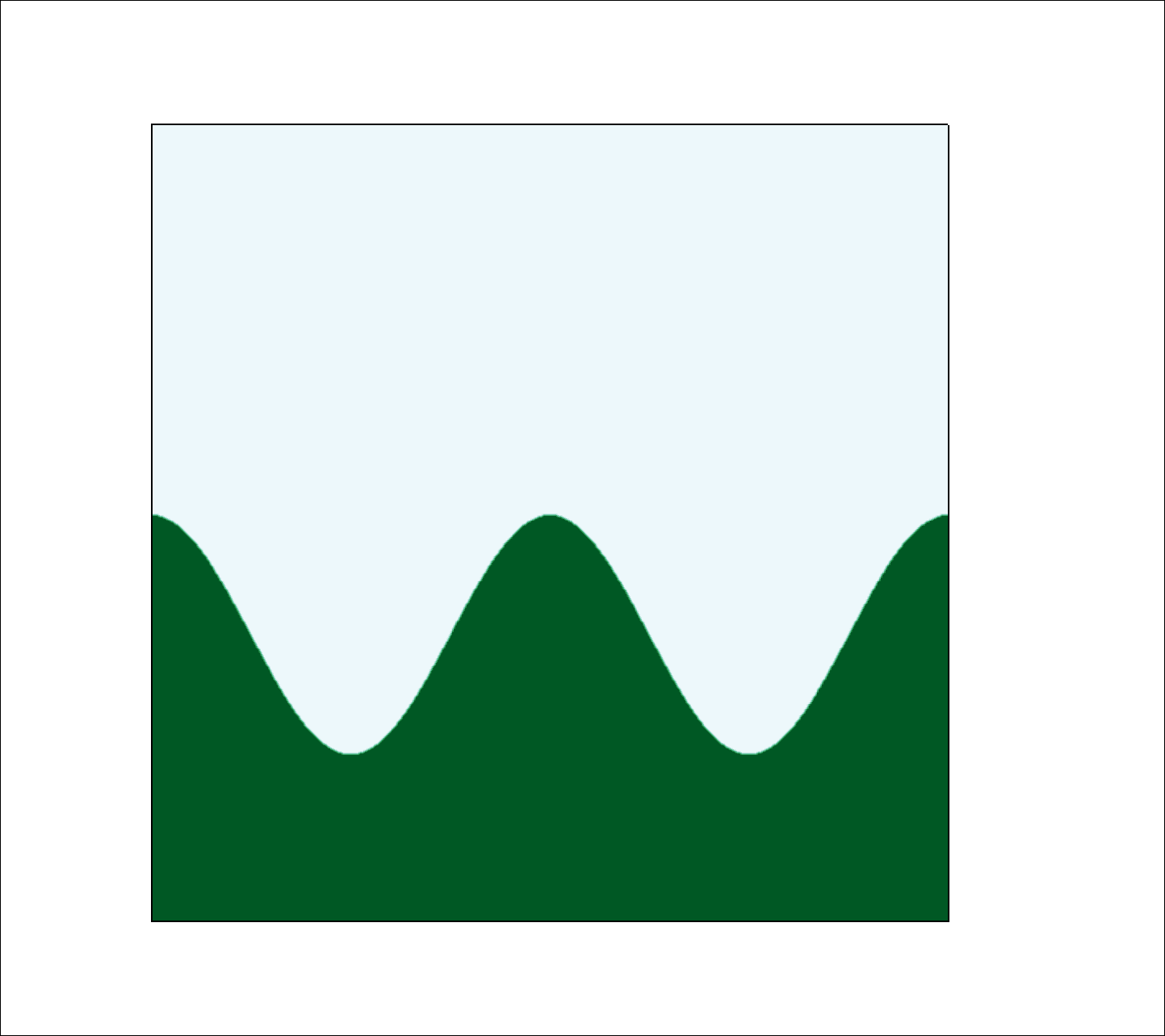}}
    \subfigure[t=146 $\mu$s]    {\includegraphics[width=.42\textwidth,trim={2.5cm 2.5cm 2.5cm 2.5cm},clip]{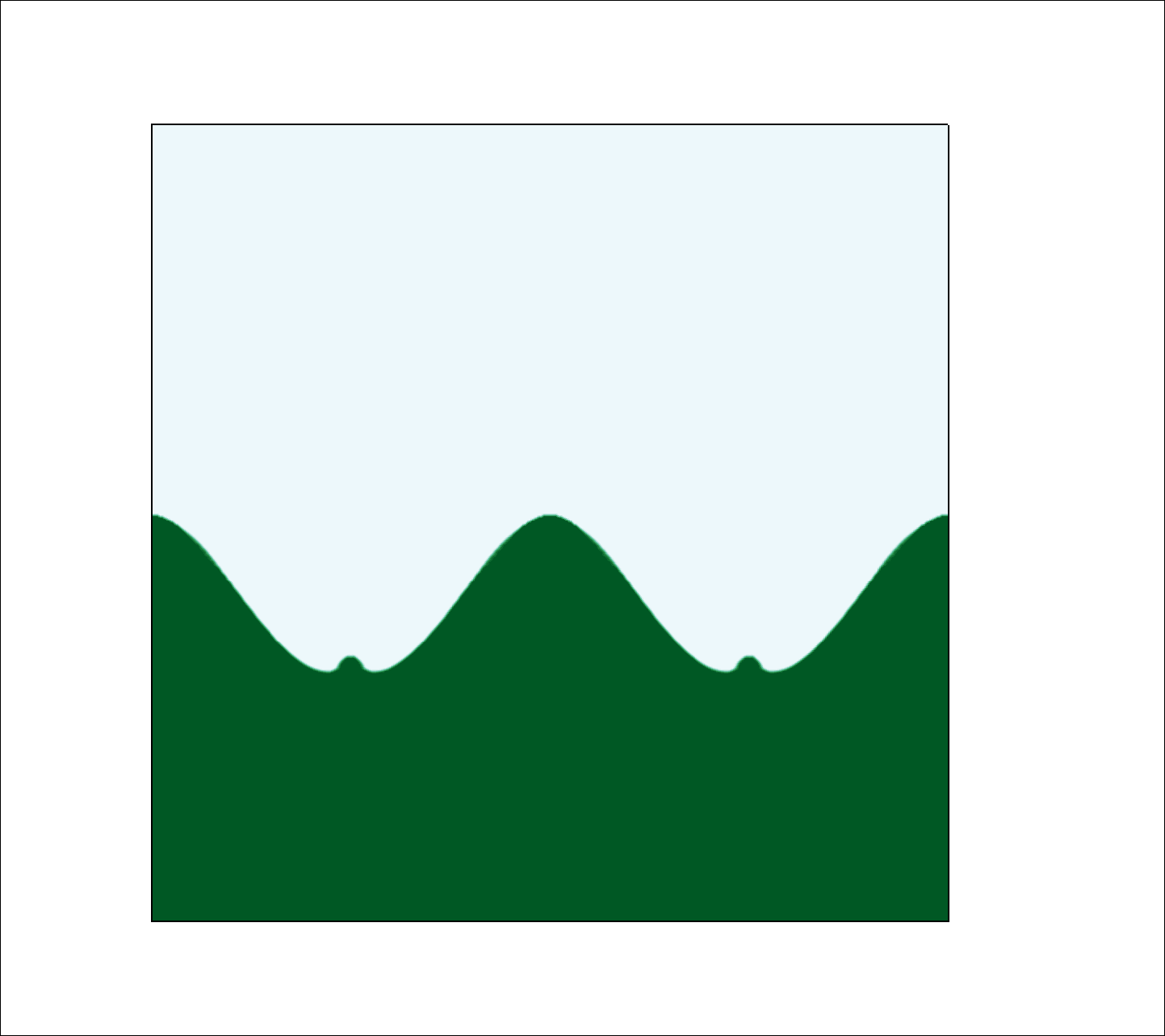}}
	\subfigure[t=202 $\mu$s]{\includegraphics[width=.42\textwidth,trim={2.5cm 2.5cm 2.5cm 2.5cm},clip]{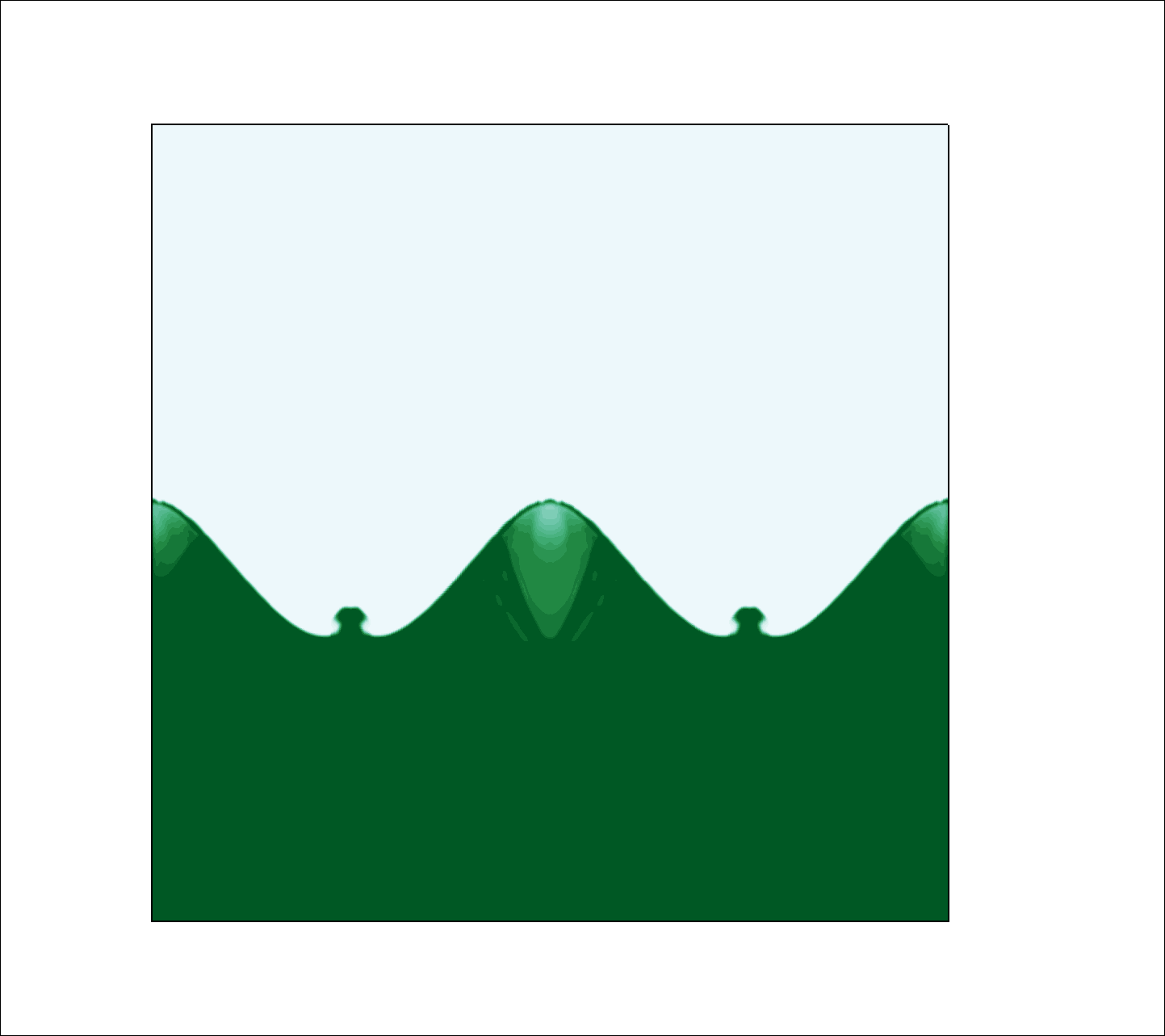}}
	\subfigure[t=290 $\mu$s]{\includegraphics[width=.42\textwidth,trim={2.5cm 2.5cm 2.5cm 2.5cm},clip]{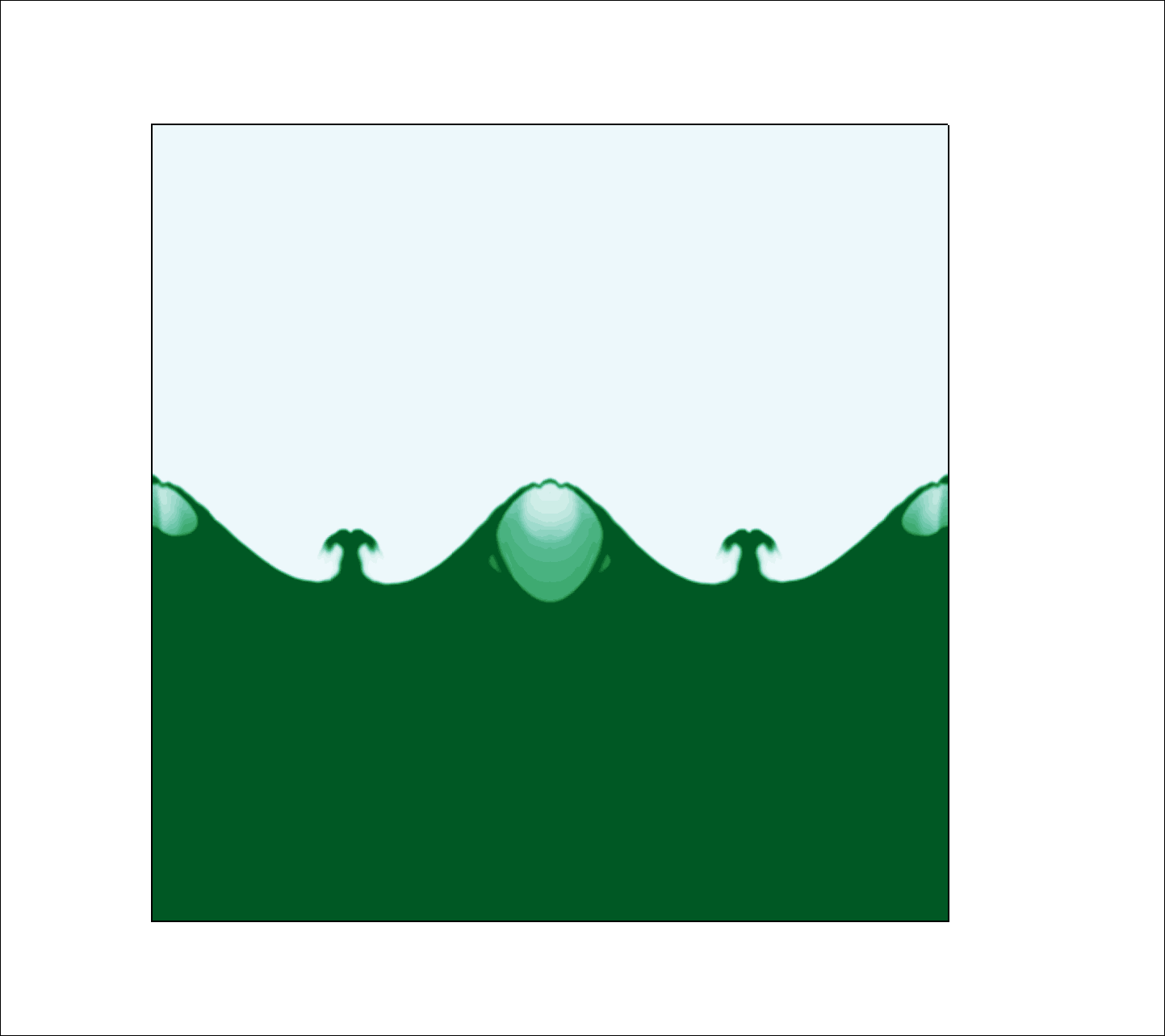}}
	\end{center}
	\protect\caption{Volume fraction of liquid water for numerical solutions of Richtmyer-Meshkov instability with cavitation at different instances. 
    \label{fig:RMinstability}}	
\end{figure}

\section{Concluding Remarks}\label{sec:concluding}
This work presents a novel hybrid approach that combines the Godunov-type scheme with the projection solution procedure to simulate multi-component flow systems characterized by the simultaneous presence of pressure-velocity coupling and pressure-density coupling dominated regions.  The proposed hybrid approach begins by splitting the inviscid flux into the advection part and the pressure part. The solution variables are first updated to their intermediate states by solving the advection part with the all-speed AUSM Riemann solver. We modified advection flux in AUSM to eliminate the pressure flux term that deteriorates the accuracy at the low Mach region. To prevent the advection flux from causing spurious velocities when surface tension is present, the pressure-velocity coupling term in advection flux is modified to ensure it vanishes at material interfaces. Then, we derive the pressure Helmholtz equation to solve the final pressure and update the intermediate states to the solution variables at the next time step. The desirable properties of the proposed solver is discussed. To accurately resolve the complex flow structures including shock waves and material interfaces without obvious numerical oscillations, a newly proposed homogenous ROUND  reconstruction strategy is employed in this work.  We simulate the incompressible multiphase benchmark tests such as rising bubbles and compressible multiphase problems such as shock-droplet interactions. By including phase-transition model, we are able to simulate nucleate boiling at low-Mach nubmers and the Richtmyer-Meshkov instability with cavitation at high-speed. Thus, this work proposes a numerical solver capable of accurately simulating all-Mach-number multicomponent flows involving multiple physical processes.

In our future work, the AMR (adaptive mesh refinement) technique will be applied in the current solver, as the assumption of the four-equation model requires a fine resolution of the thermal boundary layer across the material interface. Then, the proposed solver will be applied to simulate and investigate the lithium boiling tank. 
 
\section*{Acknowledgements}
XD and OKM acknowledge the funding provided by the Engineering and Physical Sciences Research Council and First Light Fusion through the ICL-FLF
Prosperity Partnership (grant number EP/X025373/1).  Dr. Abd Essamade Saufi from FLF is acknowledged for the fruitful discussion.

\clearpage{}
\bibliographystyle{elsarticle-num-names.bst}
\bibliography{reference.bib}

\end{document}